\documentclass[article]{tlp}
\usepackage{amsmath}


\usepackage[latin1]{inputenc}
\usepackage{tikz}
\usetikzlibrary{shapes,arrows}

\tikzstyle{decision} = [diamond, draw, fill=blue!20, 
    text width=4.5em, text badly centered, node distance=3cm, inner sep=0pt]
\tikzstyle{block} = [rectangle, draw, fill=blue!20, 
    text width=5em, text centered, rounded corners, minimum height=4em]
\tikzstyle{line} = [draw, -latex']
\tikzstyle{cloud} = [draw, ellipse,fill=red!20, node distance=3cm,
    minimum height=2em]

\usepackage{setspace}

\usepackage{listings}

\usepackage{courier}
\usepackage{hyperref}
\usepackage{bookmark}
\usepackage[T1]{fontenc}
\usepackage{latexsym}
\usepackage{color}
\usepackage{todonotes}
\usepackage{epsfig}
\usepackage{picinpar}
\usepackage{amsfonts}
\usepackage{verbatim}
\usepackage{url}
\usepackage{xspace}
\usepackage{multirow}

\usepackage{times}
\usepackage{helvet}
\usepackage{courier}

\usepackage{color}

\usepackage{multicol}
\usepackage{natbib}
\def \cite#1{\citep{#1}}
\def \shortcite#1{(\citeyear{#1})}

\def\clingcon{{\sc clingcon}\xspace}
\def\ezcsp{{\sc ezcsp}\xspace}

\def\zThree{{\sc z3}\xspace}

\def\dingo{{\sc dingo}\xspace}
\def\mingo{{\sc mingo}\xspace}

\def\lr{\hbox{lr}\xspace}

\def\bN{\mathbb{N}}

\newcommand{\exref}[1]{Example \ref{#1}}

\usepackage{todonotes}
\usepackage{times,helvet,courier}
\usepackage{amsmath,amssymb}
\usepackage{latexsym}
\usepackage{xspace}
\usepackage{verbatim}
\usepackage{url}
\usepackage{color}
\usepackage{paralist}
\usepackage{booktabs}
\usepackage{alltt}
\usepackage{caption}
\usepackage{enumitem}

\newtheorem{example}{Example}
\newtheorem{lemma}{Lemma}

\renewcommand{\paragraph}[1]{
	
	\medskip
	\noindent 
	{\bf #1}~}
\def\beq{\begin{equation}}
\def\eeq#1{\label{#1}\end{equation}}
\def\ba{\begin{array}}
\def\ea{\end{array}}

\def\<{\langle}
\def\>{\rangle}

\def\bs{\backslash}

\def\F{\mathcal{F}}

\def\cB{\ensuremath{\mathcal{B}}}

\def\cF{\ensuremath{\mathcal{F}}}

\def\cL{\ensuremath{\mathcal{L}}}
\def\cI{\ensuremath{\mathcal{I}}}

\def\cR{\mathcal{R}}

\def\bR{\mathbb{R}}
\def\cA{\mathcal{A}}

\def\bZ{\mathbb{Z}}

\newcommand{\C}{\mathcal{C}}

\newcommand{\ignore}[1]{}

\newcommand{\var}[1]{\ensuremath{{#1}_{|v}}}
\newcommand{\rels}[1]{\ensuremath{{#1}_{|r}}}
\newcommand{\fsym}[1]{\ensuremath{{#1}_{|f}}}

\def\true{\mathit{true}}
\def\false{\mathit{false}}

\def\ar{\leftarrow}
\def\rar{\rightarrow}

\def\beq{\begin{equation}}
\def\eeq#1{\label{#1}\end{equation}}

\newcommand{\At}{\mathit{At}}

\newcommand{\hd}{\mathit{hd}}

\newtheorem{definition}{Definition} 
\newtheorem{theorem}{Theorem} 

\setlist[itemize]{nosep}

\begin{document}
	%

\setlength{\abovedisplayskip}{10.0pt plus 2.0pt minus 5.0pt}
\setlength{\belowdisplayskip}{10.0pt plus 2.0pt minus 5.0pt}
\setlength{\abovedisplayshortskip}{0.0pt plus 3.0pt}
\setlength{\belowdisplayshortskip}{6.0pt plus 3.0pt minus 3.0pt}

\newcommand{\tcb}[1]{\textcolor{blue}{#1}}

\title[On Relation between CASP and SMT]{On Relation between Constraint Answer Set Programming and Satisfiability Modulo Theories\footnote{This is an extended version of the paper that appeared at IJCAI-2016~\cite{sus16}.}}

\author[Yuliya Lierler and Benjamin Susman]{YULIYA LIERLER and BENJAMIN SUSMAN \\
Department of Computer Science, University of Nebraska at Omaha \\
Omaha, Nebraska 68182, USA \\
\email{ylierler@unomaha.edu, bsusman@unomaha.edu}
}


\maketitle

\label{firstpage}

\begin{abstract}
Constraint answer set programming is a promising research direction that integrates answer set programming with constraint processing. It is often informally related to the field of satisfiability modulo theories. Yet, the exact formal link is obscured as the terminology and concepts used in these two research areas differ. In this paper, we connect these two research  areas by uncovering the precise formal relation between them. We believe  that this work will booster the cross-fertilization of  the theoretical foundations and the existing solving methods in both areas. As a step in this direction we provide a translation from constraint answer set programs with integer linear constraints to satisfiability modulo linear integer arithmetic that paves the way to utilizing modern satisfiability modulo theories solvers for computing answer sets of constraint answer set programs. \\Under consideration in Theory and Practice of Logic Programming (TPLP).

\bigskip
\noindent
{\em KEYWORDS:}
constraint answer set programming, constraint satisfaction processing, satisfiability modulo theories

\end{abstract}

\section{Introduction}

 Constraint answer set programming (CASP)~\cite{elk04,Mellarkod2009,geb09,bal09a,lier14} is a promising 
research direction that integrates answer set programming, a powerful knowledge representation paradigm, with constraint processing.   
Typical answer set programming tools start their computation with grounding, a process that substitutes variables for passing constants in respective domains. Large domains often form an obstacle for classical answer set programming. 
CASP enables a mechanism to model constraints over large domains so that they are processed in a non-typical way for answer set programming tools by delegating their solving  to constraint solver systems specifically designed to handle large and sometimes infinite domains.
CASP solvers including 
{\sc clingcon}
~\cite{geb09} and 
{\sc ezcsp} 
~\cite{bal09a} already put  CASP on the map of efficient automated reasoning tools.

Constraint answer set programming often cites itself as a related initiative to satisfiability modulo theories~(SMT) solving~\cite{BarTin-14}. Yet, the exact link is obscured as the terminology and concepts used in both fields differ.
To add to the complexity of the  picture several 
answer set programming modulo theories formalisms have been proposed. For instance,
Liu~et al.~\shortcite{liu12},  
Janhunen~et al.~\shortcite{jan11}, and
Lee and Meng~\shortcite{lee13}
introduced logic programs modulo linear constraints,
logic programs modulo difference constraints, and ASPMT programs respectively. Acyclicity programs~\cite{bom15} (or logic programs modulo acyclicity constraints) form another recently investigated formalism that parallels satisfiability modulo graphs framework developed within SMT~\cite{geb14}. 

This work attempts to unify the terminology used in CASP and SMT so that the differences and similarities of logic programs with constraints versus logic programs modulo theories become apparent. At the same time, we introduce the notion of constraint formulas, which is similar to that of logic programs with constraints. 
We identify a special class of SMT theories that we call ``uniform''. Commonly used theories in satisfiability modulo solving such as  integer linear, difference logic, and linear arithmetics
belong to uniform theories.
This class of theories helps us to establish precise links (i)~between CASP and SMT, and (ii) between  constraint formulas and SMT formulas.
We are able to then provide a formal description relating a family of distinct constraint answer set programming formalisms.

We show that this unified outlook allows us not only to better understand the landscape of CASP languages and systems, but also to foster new ideas for the design of CASP solvers and possibly SMT solvers.
For example, theoretical results of this work establish a simple method for using SMT systems for computing answer sets of a broad class of tight constraint answer set programs. \citeauthor{sus16b}~\shortcite{sus16b} utilized this method in implementing an SMT-based solver for such programs. 
In the conclusion of this work, we rely on the concept of level ranking by~\citeauthor{nie08}~\shortcite{nie08} to develop a translation for nontight constraint answer set programs to SMT formulas so that an SMT solver can be used to compute answer sets of such programs. 

\paragraph{Paper Outline}  Section~\ref{sec:intro} is on preliminaries. It  reviews concepts of logic programs, completion,  (input) answer sets, and level ranking. Section~\ref{sec:casp} presents the details of  generalized constraint satisfaction problem and links this notion to classical constraint satisfaction. The section on preliminaries concludes with formal definitions of linear and integer linear constraints. Section~\ref{sec:caspcf} introduces
constraint answer set programs and  constraint formulas. 
Next, in Section~\ref{sec:smt} we present satisfiability modulo theories 
and specify a  class of uniform theories. 
Section~\ref{sec:smtformulas} defines 
 SMT formulas and ASPT programs. Uniform theories provide  us with a ground to establish a formal link between CASP and SMT in Section \ref{sec:relation} by relating 
 SMT formulas with constraint formulas and ASPT programs with constraint answer set programs.
Section \ref{sec:relation}  concludes by characterizing a family of distinct constraint answer set programming formalisms using the uniform terminology proposed in this work. Section~\ref{sec:nontight} utilizes the generalization of level ranking to propose a method of using SMT solvers for computing answer sets of constraint answer set programs with integer constraints. Finally, we list the conclusions.

\section{Preliminaries}\label{sec:intro}
This  section starts by reviewing logic programs and 
the concept of an answer set. It also introduces programs' completion. 
Next, the generalized constraint satisfaction problems are introduced and related to the classical constraint satisfaction problems studied in artificial intelligence.

\subsection{Logic Programs and Completion}
\paragraph{Syntax} 
A {\em vocabulary} is a set of propositional symbols also called atoms.
As customary, a {\em literal} is an atom~$a$ or its negation, denoted~$\neg a$.
A \emph{(propositional) logic program}, denoted by~$\Pi$, over vocabulary~$\sigma$  is a 
set of \emph{rules} of the form
\begin{equation}\label{e:rule}
\begin{array}{l}
a\ar b_1,\ldots, b_\ell,\ not\  b_{\ell+1},\ldots,\ not\  b_m,\ 
\ not\  \ not\  b_{m+1},\ldots,\ not\  \ not\  b_n
\end{array}
\end{equation}
where $a$ is an atom over $\sigma$ or $\bot$, and each $b_i$, $1\leq i\leq n$, 
is an atom in $\sigma$.
We will sometimes use the abbreviated form for  rule~\eqref{e:rule}
\begin{equation}\label{e:rulea:b}
\begin{array}{l}
a\ar B\ 
\end{array}
\end{equation}
where $B$ stands for $b_1,\ldots, b_\ell,\ not\  b_{\ell+1},\ldots,\ not\  b_m,\ not\  \ not\  b_{m+1},\ldots,\ not\  \ not\  b_n$ and is also called a {\em body}. Syntactically, we identify rule~\eqref{e:rule} with the propositional formula
\begin{equation}\label{e:prop-formula}
b_1\wedge\ldots\wedge b_\ell\wedge \neg  b_{\ell+1} \wedge\ldots\wedge\neg  b_m \wedge \neg\neg b_{m+1} \wedge\ldots\wedge \neg\neg  b_n \rightarrow a
\end{equation}
and $B$ with the propositional formula
\begin{equation}\label{e:body-formula}
b_1\wedge\ldots\wedge b_\ell\wedge \neg  b_{\ell+1} \wedge\ldots\wedge\neg  b_m \wedge \neg\neg b_{m+1} \wedge\ldots\wedge \neg\neg  b_n.
\end{equation}
 Note that (i) the order of terms in~\eqref{e:body-formula} is immaterial, (ii) {\emph{not} is replaced with classical negation ($\neg$)}, and (iii) comma is replaced with conjunction ($\wedge$).  
 Expression $$b_1\wedge\ldots\wedge b_\ell$$ in formula~\eqref{e:body-formula} is referred to as the {\em positive} part of the body and the remainder of~\eqref{e:body-formula} as the {\em negative} part of the body. Sometimes, we interpret semantically  rule~\eqref{e:rule} and its body as propositional formulas, in these cases it is obvious that double negation $\neg\neg$ in~\eqref{e:prop-formula} and~\eqref{e:body-formula}  can be dropped.


The expression $a$ is the \emph{head} of the rule. When $a$ is $\bot$, we often omit it and say that the head is empty. 
We write $\hd(\Pi)$ for the set of nonempty heads of rules in~$\Pi$.

We call a rule whose body is empty a {\em fact}. In such cases, we drop the arrow. We sometimes may identify a set $X$ of atoms with the set of facts $\{a. \mid a \in X\}$.

\paragraph{Semantics} 
 We say a set~$X$ of atoms {\em satisfies} rule~\eqref{e:rule}, if~$X$ satisfies the propositional formula~\eqref{e:prop-formula}, where we identify $X$ with an assignment over the atoms in ~\eqref{e:prop-formula} in a natural way: 
 \begin{itemize}
 \item  any atom that occurs in $X$ maps to truth value $\true$ and
 \item  any atom in~\eqref{e:prop-formula} but not in $X$ maps to truth value $\false$.
 \end{itemize}
 We say~$X$ satisfies a program~$\Pi$, if~$X$ satisfies every rule in~$\Pi$. In this case, we also say that $X$ is a model of $\Pi$. We may abbreviate satisfaction relation with symbol $\models$ (to denote that a set of atoms satisfies a rule or a program or a formula).

The {\sl reduct} $\Pi^X$ of a program $\Pi$ relative to a set $X$ of atoms is 
obtained by first removing all rules~\eqref{e:rule} such that $X$ does not satisfy negative part of the body~$\neg  b_{\ell+1} \wedge\ldots\wedge\neg  b_m \wedge \neg\neg b_{m+1} \wedge\ldots\wedge \neg\neg  b_n$, and replacing all remaining rules with~$a\ar b_1,\ldots, b_\ell$. 
A set~$X$ of atoms is an {\em answer~set}, if it is the minimal set that satisfies all rules of $\Pi^X$~\cite{lif99d}.

Ferraris and Lifschitz~\shortcite{fer05b} showed that a choice rule $\{a\} \ar B$ can be seen as an abbreviation for a rule $a \ar\ not\ not\ a, B$ (choice rules were introduced by \citeauthor{nie00} \shortcite{nie00} and are commonly used in answer set programming languages). We adopt this abbreviation in the rest of the paper.
 
\begin{example}\label{ex:acp} Consider the logic program  from Balduccini and Lierler~\shortcite{lierbal16}:
\begin{equation}\label{eq:acp}
\ba l
  \{switch\}.\\
  lightOn\ar\ switch, not\ am.\\
  \ar not\ lightOn. \\
  \{am\}.\\
\ea
\end{equation}
Each rule in the program can be understood as follows:
\begin{itemize}
\item The action {\em switch} is exogenous.
\item The light is on ({\em lightOn}) during the night ({\em not am}) when the action {\em switch} has occurred.
\item The light must be on.
\item It is night ({\em not am}) or morning ({\em am})
\end{itemize}
Choice rules $\{switch\}.$ and 
$\{am\}.$ in program~\eqref{eq:acp}
 abbreviate rules 
$$\ba{l}
switch\ar not\ not\ switch.\\
am\ar not\ not\ am.
\ea
$$
respectively. Consider set $\{switch, \ lightOn\}$ of atoms. 
The reduct of program~\eqref{eq:acp} relative to this set follows:
$$
\ba l
  switch. \\
  lightOn\ar\ switch.\\
\ea
$$
It is easy to see that set $\{switch, \ lightOn\}$ satisfies every rule of the reduct. Furthermore, this set is minimal among sets with this property. Thus, it is an answer set of program~\eqref{eq:acp}.
In fact, it is the only answer set of this program.
This answer set suggests that the only situation that satisfies the specifications of the problem is such that (i) it is currently night, (ii) the light has been switched on, and (iii) the light is on.
\end{example}

\paragraph{Completion}
It is customary for a given vocabulary $\sigma$, to identify a set $X$ of atoms over $\sigma$ with (i) a complete and consistent set of literals over $\sigma$ constructed as $X\cup\{\neg a \mid a\in\sigma\setminus X\}$, and respectively with (ii)~an assignment function or interpretation that assigns truth value $\true$ to every atom in~$X$ and~$\false$ to every atom in $\sigma\setminus X$. 

By $Bodies(\Pi,a)$ we denote the set of the bodies of all rules of~$\Pi$ with head~$a$.
For a program $\Pi$ over vocabulary~$\sigma$, the {\sl completion} of~$\Pi$~\cite{cla78}, denoted by $Comp(\Pi)$, is  the set of classical formulas that consists of the rules~\eqref{e:rule}  in $\Pi$ (recall that we identify rule~\eqref{e:rule} with  implication~\eqref{e:prop-formula})  and the implications
\beq
a\rar \bigvee_{a\ar B\in \Pi} B
\eeq{eq:comp2}
for all atoms $a$ in $\sigma$.
When set $Bodies(\Pi,a)$ is empty, the implication~\eqref{eq:comp2} has the form $a\rar\bot$.
When a rule~\eqref{e:rulea:b} is a fact $~a.~$, then we identify this rule  with the  clause consisting of a single atom $a$.

\begin{example}\label{ex:asp-comp}
The completion of logic program~\eqref{eq:acp} consists of formulas
\begin{equation} \label{eq:acp-comp2}
\ba l
  \neg \neg switch \rar switch, \\
  switch \wedge \neg am\rar lightOn,\\
  \neg lightOn \rar \bot,\\
  \neg \neg am \rar am,\\
  switch \rar \neg \neg switch, \\
  lightOn \rar switch \wedge \neg am, \\
  am \rar \neg \neg am.\\
\ea
\end{equation}
It is easy to see that this completion is equivalent to the set of formulas
\begin{equation}\label{eq:acp-comp}
\ba l
  lightOn \leftrightarrow switch \wedge \neg am, \\
  lightOn. \\
\ea
\end{equation}
The set $\{switch,lightOn\}$ is the only model of~\eqref{eq:acp-comp}. 
(Unless the signature of a formula is explicitly stated, we consider the set of atoms occurring in the formula to implicitly specify its signature.)
Note that set $\{switch,lightOn\}$ coincides with the only answer set of program~\eqref{eq:acp}. 
\end{example}

\paragraph{Tightness}
Any answer set of a program is also a model of its completion. The converse does not always hold.
Yet, for the large class of logic programs, called {\sl tight},  their answer sets coincide with models of their completion~\cite{fag94,erd01}.
Tightness is a syntactic condition on a program that can be verified by means of program's dependency  graph. 

The {\em dependency graph} of~$\Pi$ is the directed graph $G$ such that 
\begin{itemize}
\item the vertices of $G$ are the atoms occurring in~$\Pi$, and
\item for every rule~\eqref{e:rule} in~$\Pi$ whose head is not $\bot$, $G$ has an edge from atom $a$ to each atom in positive part $b_1,\dots, b_\ell$ of its body.
\end{itemize}
A program is called {\em tight} if its dependency graph is acyclic.

For example, the dependence graph of 
program~\eqref{eq:acp} consists of three nodes, namely, $am$, $switch$, and $lightOn$ and a single edge from $lightOn$ to $switch$. This program is obviously tight.

\paragraph{Level Rankings}
\citeauthor{nie08}~\shortcite{nie08} characterized answer sets of "normal" logic programs  in terms of "level rankings".
Normal programs  consist of rules of the form~\eqref{e:rule}, where $n=m$ and $a$ is an atom. 
Thus, such constructs as choice rules and so-called denials (rules with  empty head) are not covered by normal programs.
We generalize the concept of level ranking to  programs considered in this paper that are more general than normal ones.

We start by introducing some notation. 
By  $\mathbb{N}$ we denote the set of natural numbers.
For a rule~\eqref{e:rulea:b}, by $B^+$ we denote  its  positive part  and sometimes identify it with the set of atoms that occur in it, i.e., $\{b_1,\dots, b_l\}$ (recall that $B$ in~\eqref{e:rulea:b} stands for the right hand side of the arrow in rule \eqref{e:rule}).
For a program~$\Pi$, by $\At(\Pi)$ we denote the set of atoms  occurring in it.


\begin{definition}\label{def:lr}
	For a logic program $\Pi$ and a set $X$  of atoms over $\At(\Pi)$, a function \lr: $X \rightarrow \mathbb{N}$ is a \emph{level-ranking} of $X$ for $\Pi$ when for each $a \in X$, there is $B$ in $Bodies(\Pi,a)$ such that~$X$ satisfies~$B$  and 
	 for every $b \in B^+$ it holds that $\lr(a) - 1 \geq \lr(b)$. 
\end{definition}
\citeauthor{nie08}~\shortcite{nie08} observed that for an arbitrary normal logic programs, a model $X$ of its completion is also an answer set for this program when there is a level ranking of $X$ for the program. We generalize this result beyond normal programs.

\begin{theorem}\label{thm:ans-iff-lr}
		For a program $\Pi$ and a set $X$ of atoms that is a model of its completion $Comp(\Pi)$, $X$ is an answer set 
	of $\Pi$ if and only if  there is a 
	level ranking of $X$ for~$\Pi$. 
\end{theorem}

\begin{proof}
The proof largely follows the lines of the proof of Theorem~1 from
\citeauthor{nie08}~\shortcite{nie08} but utilizes the terminology used in this paper.
We start by defining an operator $T_\Pi(I)$ for a program $\Pi$ and a set $I$ over $\At(\Pi)\cup\{\bot\}$ as follows: $$T_\Pi(I)=\{a \mid a\ar B \in \Pi, I \hbox{ satisfies } B \}.$$ For this operator we define
	\[T_\Pi\uparrow0 = \emptyset\]
	and for $i=0,1,2,\dots$
	\[ T_\Pi\uparrow(i+1) =T_\Pi(T_\Pi\uparrow i) \]

	Left-to-right:
	Assume $X$ is an answer set of $\Pi$. We can construct a level ranking \lr  of~$X$ for $\Pi$ using the $T_{\Pi^X}(\cdot)$ operator.
	As $X$ is an answer set of $\Pi$, we know that $X = T_{\Pi^X}\uparrow \omega$ and for each $a \in X$ there is a unique $i$ such that $a \in T_{\Pi^X}\uparrow i$,  but $a \not \in T_{\Pi^X}\uparrow(i-1)$. We consider $\lr(a)=i$. We now illustrate that $\lr$ is indeed a level ranking. For $a \in X$ there is a rule $a\ar B$ of the form (\ref{e:rule}) such that $a\ar b_1,\ldots, b_\ell \in \Pi^X$ and $T_{\Pi^X}\uparrow(i-1)$ satisfies $b_1\wedge \ldots\wedge b_\ell$. Consequently, for every $b_j$ in  $\{b_1,\ldots, b_\ell\}$, 
		$lr(b_j) \leq i -1$. Thus, $\text{lr}(a) -1 \geq \text{lr}(b_j)$.
Also, from the way the reduct is constructed it follows that $X$ satisfies body $B$  of rule $a\ar B$.

		Right-to-left:
	Assume that there is a level ranking \lr of $X$ for $\Pi$.  We show that then $X$ is the minimal model of $\Pi^X$, which implies that $X$ is an answer set of $\Pi$. Since $X$ is a model of the completion of $\Pi$ it follows that $X$ is also a model of $\Pi$. Indeed, recall the construction process of completion. From the construction of $\Pi^X$ it follows that  $X$ is also a model of  $\Pi^X$. Proof by contradiction. 
	Assume that there is a set of atoms $X' \subset X$ such that $X'$ is a model of~$\Pi^X$ and hence $X$ is not a minimal model of $\Pi^X$. Consider now an atom $a \in X\setminus X'$ with the smallest level ranking $\lr(a)$. 
	Since  \lr is a level ranking of $X$ for~$\Pi$, it follows that there is a rule $a\ar B$ in $\Pi$ such that  $X \models B$ and for every $b \in B^+ $, it holds that  $\lr(b) < \lr(a)$. 
	As we considered atom $a$ in $X\setminus X'$ with the smallest level ranking it follows that $b\in X'$. By $\Pi^X$ construction, rule $a\ar B^+$ belongs to $\Pi^X$. We derive that~$X'$ satisfies $B^+$, but does not contain $a$. This contradicts out assumption that~$X'$ satisfies~$\Pi^X$ as it does not satisfy rule $a\ar B^+$.
\end{proof}


\subsection{Generalized Constraint Satisfaction Problems}\label{sec:casp}

In this section we present a primitive constraint as defined by Marriott and Stuckey~(\citeyear[Section 1.1]{mar98}). We  refer to this concept as a constraint, dropping the word ``primitive''. We use constraints to define  a generalized constraint satisfaction problem that Marriott and Stuckey refer to  as ``constraint''. 
We then review constraint satisfaction problems as commonly defined in artificial intelligence literature and illustrate that they form a  special case of  generalized constraint satisfaction problems. We finally introduce linear constraints and linear constraint satisfaction problems.

\paragraph{Signature, c-vocabulary, constraint atoms}
We adopt the following convention: for a function~$\nu$ and an element~$x$, by $x^\nu$ we denote 
the value that function $\nu$ maps~$x$ to, in other words, $x^\nu=\nu(x)$.
A {\em domain} is a {\em nonempty} set of elements (or values).
A {\em signature}~$\Sigma$ is a set of 
 {\em variables}, 
{\em predicate symbols}, 
and 
 {\em function symbols (or f-symbols)}.
Predicate and function symbols are associated with a positive integer called {\em arity}.
By $\var{\Sigma}$, $\rels{\Sigma}$, and~$\fsym{\Sigma}$  we denote the subsets of $\Sigma$ that contain all variables, all predicate symbols, and all f-symbols respectively.

For instance, we can define  signature 
$\Sigma_1=\{s,r,E,Q\}$
by saying that $s$ and $r$ are variables,~$E$ is a predicate symbol of arity $1$, and~$Q$ is a predicate symbol of arity $2$. 
Then, $\var{{\Sigma_1}}=\{s,r\}$, $\rels{{\Sigma_1}}=\{E,Q\}$, $\fsym{{\Sigma_1}}=\emptyset$.

Let $D$ be a domain.
For a 
  set $V$ of variables, 
we call a total function $\nu:V\rightarrow D$ a \hbox{{\em $[V,D]$ valuation}}.
For a set~$R$ of predicate symbols,
we call a total function on~$R$ {\em an $[R,D]$ r-denotation}, when
it maps each $n$-ary predicate symbol of $R$ into an $n$-ary relation on $D$.
For a set $F$ of f-symbols,
we call a total function on $F$ {\em an $[F,D]$ f-denotation}, when
it maps each $n$-ary f-symbol of $F$ into a function $D^{n}\rightarrow D$. 

Table~\ref{tbl:cspDefinitions} presents sample definitions of  a domain, valuations, and r-denotations. In the remainder of the paper we frequently refer to these sample   valuations, and r-denotations.

\begin{table}[h]
\caption{Example definitions for signature, valuation, and r-denotation}\label{tbl:cspDefinitions}
\vspace{-2mm}
 \begin{tabular}{ll}
 \hline \hline 
 $\Sigma_1$ ~~~~~ & $\{s,r,E,Q\}$\\
 $D_1$ & $\{1,2,3\}$\\
 $\nu_1$ &  $[\var{\Sigma_1},D_1]$ valuation, where $s^{\nu_1}=1$ and $r^{\nu_1}=1$\\
$\nu_2$ & $[\var{\Sigma_1},D_1]$ valuation, where $s^{\nu_2}=2$ and $r^{\nu_2}=1$\\
$\rho_1$ & $[\rels{\Sigma_1},D_1]$ r-denotation, where
$
E^{\rho_1}=\{\langle 1 \rangle\},~~~~~~~
Q^{\rho_1}=\{\langle 1,1 \rangle,\langle 2,2 \rangle,\langle 3,3 \rangle\}$\\
$\rho_2$ &  $[\rels{\Sigma_1},D_1]$ r-denotation, where
$
E^{\rho_2} = \{\langle 2 \rangle,\langle 3 \rangle\}$, $Q^{\rho_2}=Q^{\rho_1}$.\\
\hline \hline
 \end{tabular}
\end{table}

A {\em constraint vocabulary (c-vocabulary)} is a pair $[\Sigma,D]$, where~$\Sigma$ is a signature and~$D$ is a domain.
A {\em term} over  a c-vocabulary~$[\Sigma,D]$ is either 
\begin{itemize}
	\item a variable in $\var{\Sigma}$,
	\item a domain element in $D$, or  
	\item an expression $f(t_1,\dots,t_n)$, where $f$ is an f-symbol of arity $n$ in  $\fsym{\Sigma}$ 
and $t_1,\dots,t_n$ are terms over $[\Sigma,D].$
\end{itemize}
A {\em  constraint atom} over a c-vocabulary $[\Sigma,D]$  is 
 an expression 
 \beq
 P(t_1,\dots,t_n),
 \eeq{eq:catom}
 where $P$ is a predicate symbol from $\rels{\Sigma}$ of arity~$n$  and~$t_1,\dots,t_n$ are terms over~$[\Sigma,D]$. 
 A {\em  constraint literal} over a c-vocabulary $[\Sigma,D]$  is either a constraint atom~\eqref{eq:catom} or
 an expression 
 \beq
 \neg P(t_1,\dots,t_n),
 \eeq{eq:clit}
 where $P(t_1,\dots,t_n)$ is a constraint atom over $[\Sigma,D]$.
 
For instance, expressions
$$\neg E(s),~~ \neg E(2),~~ Q(r,s)$$
are constraint literals over c-vocabulary $[\Sigma_1,D_1]$, where $\Sigma_1$ and $D_1$ are defined in Table~\ref{tbl:cspDefinitions}.

It is due to notice that syntactically, constraint literals are similar to {\em ground} literals of predicate logic. (In predicate logic, variables as defined here are referred to as {\em object constants} or {\em function symbols of  arity $0$}.) The only difference is that here domain elements are allowed to form a term. For instance, an expression
$E(2)$ is a
 constraint atom over~$[\Sigma_1,D_1]$, where $2$ is a term formed from a domain element. 
 In predicate logic, domain elements are not part of a signature over which atoms are formed.

We now proceed to introducing satisfaction relation  for constraint literals. Let  $[\Sigma,D]$ be a   c-vocabulary,
$\nu$ be a $[\var{\Sigma},D]$  valuation,
$\rho$ be a $[\rels{\Sigma},D]$ r-denotation,
and~$\phi$ be a $[\fsym{\Sigma},D]$ f-denotation. 
First, we  define recursively  
 a value that valuation $\nu$ assigns to a term $\tau$ over $[\Sigma,D]$ with respect to $\phi$.
 We denote this value by~$\tau^{\nu,\phi}$ and compute it  as follows:
 \begin{itemize}
 	\item for a term that is a variable $x$ in $\var{\Sigma}$, $x^{\nu,\phi}=x^\nu$,
 	\item for a term that is a domain element $d$ in $D$, $d^{\nu,\phi}$ is $d$ itself,
  	\item for a term $\tau$ of the form~$f(t_1,\dots,t_n)$, $\tau^{\nu,\phi}$ is defined recursively by the formula
  \[
  f(t_1,\dots,t_n)^{\nu,\phi}=f^{\phi}(t_1^{\nu,\phi},\dots,t_n^{\nu,\phi}).
  \]
 \end{itemize}
Second, we define what it means for valuation to be a solution of a constraint literal with respect to given r- and f-denotations.
We say that $\nu$ {\em satisfies (is a solution to)} constraint literal~\eqref{eq:catom} over~$[\Sigma,D]$ {\em with respect to~$\rho$ and~$\phi$} when
$\langle  t_1^{\nu,\phi},\dots,t_n^{\nu,\phi}\rangle\in P^{\rho}$.
Let $\cR$ be an n-ary relation on~$D$. By $\overline{\cR}$ we denote \emph{complement}  relation of $\cR$ constructed as $D^n\setminus \cR$.
Valuation $\nu$ {\em satisfies (is a solution to)} constraint literal of the form~\eqref{eq:clit}
{\em with respect to~$\rho$ and~$\phi$} when
$\langle  t_1^{\nu,\phi},\dots,t_n^{\nu,\phi}\rangle\in \overline{P^{\rho}}.$
 
For instance, consider declarations  of valuations~$\nu_1$ and $\nu_2$, and r-denotations~$\rho_1$ and~$\rho_2$
in Table \ref{tbl:cspDefinitions}. Valuation~$\nu_1$ satisfies constraint literal~\hbox{$Q(r,s)$}  with respect to~$\rho_1$, while valuation~$\nu_2$ does not satisfy this constraint literal with respect to~$\rho_2$. 
 
\paragraph{Lexicon, constraints, generalized constraint satisfaction problem}
We are now ready to define constraints, their syntax and semantics. 
To begin we  introduce a {\em lexicon}, which is a tuple 
 $([\Sigma,D], \rho, \phi)$, where $[\Sigma,D]$  is a c-vocabulary, $\rho$
 is a $[\rels{\Sigma},D]$ r-denotation, and $\phi$ is a $[\fsym{\Sigma},D]$ f-denotation.
 For a lexicon $\cL=([\Sigma,D], \rho, \phi)$,
 we call any function that is $[\var{\Sigma},D]$  valuation, a {\em valuation over} $\cL$.  
 We omit the last element of the lexicon tuple if the signature $\Sigma$ of the lexicon contains no f-symbols.
 A {\em constraint} is defined over lexicon $\cL=([\Sigma,D], \rho, \phi)$. 
 Syntactically, it is a constraint literal  over $[\Sigma,D]$ (lexicon $\cL$, respectively).  
 Semantically, we say that   valuation $\nu$ over~$\cL$
 {\em satisfies (is a solution to)}  the constraint $c$ when  $\nu$ satisfies~$c$
 with respect to~$\rho$ and~$\phi$.

For instance, Table~\ref{tbl:sampleLexiconsConstraints} presents definitions of  sample lexicons $\cL_1$, $\cL_2$, and  constraints $c_1$, $c_2$, $c_3$, and~$c_4$ using the earlier declarations from Table~\ref{tbl:cspDefinitions}.
Valuation $\nu_1$ from Table \ref{tbl:cspDefinitions} is a solution to $c_1$, $c_2$, $c_3$, but not a solution to~$c_4$. 
Valuation~$\nu_2$ from Table \ref{tbl:cspDefinitions} is not a solution to $c_1$,~$c_2$,~$c_3$, and $c_4$. In fact, constraint $c_4$ has no solutions.
 We sometimes omit the explicit mention of the lexicon  when talking about constraints:  we then may identify a constraint with its syntactic form of a constraint literal.
 
 \begin{table}[h]
 \caption{Sample lexicons and constraints}\label{tbl:sampleLexiconsConstraints}
  \begin{tabular}{ll}
  \hline \hline 
   $\cL_1$ & $([\Sigma_1,D_1],\rho_1)$\\
   $\cL_2$ & $([\Sigma_1,D_1],\rho_2)$\\
    $c_1$ ~~~~~~~~~~~& a literal~$Q(r,s)$ over lexicon $\cL_1$ \\
    $c_2$ & a literal~$Q(r,s)$ over lexicon $\cL_2$ \\
    $c_3$ & a literal $\neg E(s)$ over lexicon~$\cL_2$\\
   $c_4$ & a literal $\neg E(2)$ over lexicon~$\cL_2$.\\
  \hline \hline 
  \end{tabular}
 \end{table}


\begin{definition}
A {\em generalized constraint satisfaction problem~(GCSP)} $\C$ is a finite set of constraints 
  over a lexicon~$\cL=([\Sigma,D],\rho,\phi)$. 
 We say that a valuation~$\nu$ over~$\cL$
 {\em satisfies (is a solution to)} the GCSP~$\C$ when~$\nu$ is a solution to every constraint in~$\C$.  
 \end{definition}
 
 	 For example, consider a set~$\{c_2, c_3, c_4\}$ of constraints. 
 	 Any subset of this set forms a GCSP, including subsets $\{c_2, c_3\}$ and~$\{c_2, c_3, c_4\}$.	
 	 Sample valuation $\nu_1$ over lexicon~$\cL_2$ (where $\nu_1$  stems from Tables~\ref{tbl:cspDefinitions}) satisfies the GCSP $\{c_2, c_3\}$, but does not satisfy the GCSP~$\{c_2, c_3, c_4\}$.

\paragraph{From GCSP to Constraint Satisfaction Problem}
We now define a constraint satisfaction problem (CSP) as customary in classical literature on artificial intelligence. We then explain in which sense generalized constraint satisfaction problems generalize CSPs. 

We say that a lexicon is {\em finite-domain} if it is defined over a c-vocabulary that refers to a domain whose set of elements is finite. Trivially, lexicons~$\cL_1$ and~$\cL_2$ defined in Table~\ref{tbl:sampleLexiconsConstraints} are finite-domain lexicons.
Consider a special case of a constraint of the form~\eqref{eq:catom} over finite-domain lexicon $\cL=([\Sigma,D],\rho)$, so that each $t_i$ is a variable. (For instance, constraints~$c_1$, $c_2$, and~$c_3$ satisfy the stated requirements, while $c_4$ does not.)
In this case, we can  identify constraint~\eqref{eq:catom} over  $\cL$ with the pair
\beq
 \langle(t_1,\dots,t_n),P^\rho\rangle.
\eeq{eq:constr}
A {\em constraint satisfaction problem} (CSP) is a set of pairs~\eqref{eq:constr}, where $\var{\Sigma}$ and $D$ of the finite-domain lexicon~$\cL$ are called the variables  and the domain of CSP, respectively.
Saying that  valuation~$\nu$ over~$\cL$ satisfies~\eqref{eq:catom} is the same as saying that~$\langle  t_1^{{\nu}},\dots,t_n^{{\nu}}\rangle\in P^{\rho}.$
The latter is the way in which  a solution to expressions~\eqref{eq:constr} in CSP is typically defined. As in the definition of semantics of GCSP, a valuation is a {\em solution} to a CSP problem $C$ when it is a solution to every pair~\eqref{eq:constr} in $C$.

In conclusion, GCSP generalizes  CSP by 
\begin{itemize}
	\item elevating the finite-domain restriction, and
	\item allowing us more elaborate syntactic expressions (e.g., recall f-symbols).
\end{itemize}

\subsubsection{Linear and Integer Linear Constraints}\label{sec:licon}
We now define ``numeric'' signatures and lexicons and
introduce a set of constraints referred to as linear, which are commonly used in practice.
A {\em numeric signature} is a signature that satisfies the following requirements
\begin{itemize}
	\item its only predicate symbols are $<$, $>$, $\leq$, $\geq$, $=$, $\neq$  of arity 2, and
	\item its only  f-symbols are $+$, $\times$ of arity $2$. 
\end{itemize}
We use the symbol $\cA$ to denote a  numeric signature.

Symbols $\bZ$  and $\bR$ denote the sets of integers and real numbers respectively. 
Let $\rho_\bZ$ and $\phi_\bZ$ be $[\{<,>,\leq,\geq,=,\neq\},\bZ]$ r-denotation and $[\{+,\times\},\bZ]$ f-denotation respectively, where they map their predicate and function  symbols into usual arithmetic  relations and operations over integers.
Similarly, $\rho_\bR$ and  $\phi_\bR$ denote $[\{<,>,\leq,\geq,=,\neq\},\bR]$ r-denotation and $[\{+,\times\},\bR]$ f-denotation respectively, defined over the reals. We can now define the following lexicons
\begin{itemize}
	\item an {\em integer lexicon} of the form $([\cA,\bZ],\rho_\bZ,\phi_\bZ)$,
	\item a {\em numeric lexicon} of the form $([\cA,\bR],\rho_\bR, \phi_\bR)$.
\end{itemize}

A {\em (numeric) linear expression} 
has the form
\beq
a_1 x_1 +\cdots + a_n x_n, 
\eeq{eq:exp}
where $a_1,\dots,a_n$ are real numbers and $x_1,\dots,x_n$ are  variables over real numbers. 
When $a_i=1$ ($1\leq i\leq n$) we may omit it from the expression.
We view expression~\eqref{eq:exp} as an abbreviation for the following term
\[
+(\times(a_1,x_1),+(\times(a_2,x_2),\dots +(\times(a_{n-1},x_{n-1}),\times(a_n,x_n))\dots),
\]
 over some c-vocabulary $[\cA,\bR]$, where $\cA$ contains $x_1,\dots,x_n$ as its variables.
For instance, $2 x_2+ 3 x_3$ is an abbreviation for the expression 
$+(\times(2,x_2),\times(3,x_3)).$

An {\em integer linear expression} has the form~\eqref{eq:exp},  where $a_1,\dots,a_n$ are integers, and $x_1,\dots,x_n$ are variables over integers. 

We call a constraint  {\em linear (integer linear)} when it is defined  over some numeric (integer) lexicon 
and
has the form
\beq
\bowtie(e,k)
\eeq{eq:lc}
where $e$ is a linear (integer linear) expression, $k$ is a real number (an integer), and $\bowtie$ belongs to $\{<,>,\leq,\geq,=,\neq\}$.  We can write~\eqref{eq:lc} as an expression in usual infix notation $e\bowtie k$.

We call a GCSP a {\em (integer) linear constraint satisfaction problem}  when it is composed of (integer) linear constraints. 
For instance, consider integer linear constraint satisfaction problem
composed of two constraints  $x>4$ and $x<5$ (here signature $\cA$ is implicitly defined by restricting its variable to contain $x$). 
It is easy to see that this problem has no solutions. 
On the other hand,  linear constraint satisfaction problem
composed of the same two constraints  $x>4$ and $x<5$  has an infinite number of solutions, including valuation that assigns $x$ to $4.1$.


\section{Constraint Answer Set Programs and Constraint Formulas}\label{sec:caspcf}
In this section we introduce constraint answer set programs,
which merge the concepts of logic programming and generalized constraint satisfaction problems.
We also present a similar concept of  ``constraint formulas'', which merges the concepts of propositional formulas and generalized constraint satisfaction problems.
First, we introduce input answer sets, followed by constraint answer set programs and input completion. Second, we present constraint formulas. 
Finally, we demonstrate the close connection between the constraint answer set programs and constraint formulas using the concept of input completion and tightness condition.



\subsection{Constraint Answer Set Programs}\label{sec:casp2}
We start by introducing  a concept in spirit of an input answer set by Lierler and Truszczynski~\shortcite{lt2011}.\footnote{Similar concepts to input answer sets have been noted by \citeauthor{gel96}~\shortcite{gel96}, \citeauthor{oik06}~\shortcite{oik06}, and \citeauthor{den07}~\shortcite{den07}).} In particular,
we consider input answer sets ``relative to input vocabularies''.
We then extend the definition of completion and restate the result by Erdem and Lifshitz~\shortcite{erd01} for the case of input answer sets. 
The concept of an input answer set is essential for introducing constraint answer set programs.  
Constraint answer set programs (and constraint formulas) are defined over two disjoint vocabularies so that atoms stemming from those vocabularies ``behave'' differently. Input answer set semantics allows us to account for these differences.

\begin{definition}
	\label{def:input-answer-set}	
For a logic program $\Pi$ over vocabulary $\sigma$,
a set~$X$ of atoms over $\sigma$ is an \emph{input answer set} of $\Pi$ relative to vocabulary $\iota\subseteq\sigma$ so that $\hd(\Pi)\cap\iota=\emptyset$  when $X$ is an answer set of the program 
$$
\Pi\cup (X\cap\iota).
$$
\end{definition}

To illustrate the concept of an input answer set consider  program
\[
	\ba l
	lightOn\ar\ switch, not\ am.\\
	\ar not\ lightOn. \\
	\ea
\]
This program has a unique input answer set $\{switch,lightOn\}$ relative to input vocabulary $\{switch,am\}$.

Let $\sigma_r$ and $\sigma_i$ be two disjoint vocabularies.
We refer to their elements as \emph{regular}  and \emph{irregular} atoms respectively.

\begin{definition}
A {\em constraint answer set 
program} (CAS program) 
over the vocabulary $\sigma=\sigma_r
\cup\sigma_i$ is a triple $\langle \Pi,\cB,\gamma\rangle$, where 
\begin{itemize}
\item $\Pi$ is a logic
program over the vocabulary $\sigma$ such that $\hd(\Pi)\cap\sigma_i=
\emptyset$, 
\item $\cB$ is a set of constraints over some lexicon~$\cL$, 
and 
\item $\gamma$ is an injective 
function from the set $\sigma_i$ of 
irregular atoms  to the set $\cB$ of constraints.
\end{itemize}

For a CAS program $P=\langle \Pi,\cB,\gamma\rangle$ over the vocabulary $\sigma=
\sigma_r\cup\sigma_i$ so that $\cL$ is the lexicon of the constraints in $\cB$, a set $X\subseteq\At(\Pi)$ 
is an \emph{answer set}
of $P$ if
\begin{enumerate} 
\item  $X$ is an input answer set of $\Pi$ relative to $\sigma_i$, and 
\item the following GCSP over $\cL$ has a solution
\beq \{\gamma(a) \mid a\in X\cap\sigma_i\}\cup
\{\neg{\gamma(a)} \mid a\in  (\At(\Pi)\cap \sigma_i)\setminus X\}.
\eeq{eq:GCSPASP}
\end{enumerate}
Note that $\neg{\gamma(a)}$ may result in expression of the form $\neg\neg P(t_1,\dots,t_n)$ that we identify with $P(t_1,\dots,t_n)$. (We use this convention across the paper.)
\end{definition}
These definitions are  generalizations of CAS programs introduced by Gebser et al.~\shortcite{geb09} as they 
\begin{itemize}
\item 
 refer to the concept of GCSP in place of CSP in the original definition, and 
 \item  allow for more general syntax of logic rules (e.g. choice rules are covered by the presented definition). 
\end{itemize}

 This is a good place to note a restriction $\hd(\Pi)\cap\sigma_i=\emptyset$ on the form of the rules in a CAS program. This restriction states that an irregular atom may not form a head of a rule. There is a body of research including, for example, work by~\citeauthor{lee12}~(\citeyear{lee12}), ~\citeauthor{lif12a}~(\citeyear{lif12a}), and~\citeauthor{bal13a}~(\citeyear{bal13a}) that investigates variants of semantics of {\em CAS-like} programs, where analogous to irregular atoms are allowed in the heads. Lifting this restriction proves to be nontrivial.
Traditional CASP systems, such as~\clingcon or \ezcsp, restrict their attention to programs discussed here.

It is  worthwhile to remark on the second condition in the definition of an answer set for CAS programs. This condition  requires the {\em existence} of a solution to a constructed GCSP  problem, and ignores a form of a particular solution to this GCSP or a possibility of  multiple solutions. We now define a concept of an  extended answer set that  takes a solution to a constructed GCSP problem into account:
\begin{definition}
For a CAS program $P=\langle \Pi,\cB,\gamma\rangle$ over the vocabulary $\sigma=
\sigma_r\cup\sigma_i$ so that $\cL$ is the lexicon of the constraints in $\cB$, a set $X\subseteq\sigma$ and valuation $\nu$ from  variables in the signature of $\cL$ to the domain of $\cL$, a pair $\langle X,\nu\rangle$ is an   
 \emph{extended answer set}
of $P$ if $X$ is an answer set of $P$ and $\nu$ is a solution to GCSP~\eqref{eq:GCSPASP}. 
\end{definition}
CASP systems, such as~\clingcon or \ezcsp, allow the user to select whether he is interested in computing answer sets or extended answer sets of a given CAS program. In the rest of the paper we focus on the notion of an answer set, but generalizing  concepts and results introduced later to the notion of extended answer set is not difficult.

In the sequel we adopt the following notation. 
To distinguish irregular atoms from the constraints 
to which these atoms  are mapped,  
we use bars to denote that an expression  is an irregular atom. For instance, $|x<12|$ and  $|x\geq 12|$ denote
irregular atoms. Whereas inequalities (constraints) $x<12$ and $x\geq 12$, respectively,  provide natural mappings for these atoms. We assume such natural mappings  in this presentation.

\begin{example}\label{ex:casp} Let us  consider sample constraint answer set program.
We first define the integer lexicon~$\cL_3$:
\begin{equation}\label{eq:casp-lexicon}
 \begin{oldtabular}{ll}
 \hline 
  $\Sigma_2$ ~~~~~~~~~~ & a numeric signature containing one variable $\{x\}$ \\
  $\cL_3$ & $([\Sigma_2,\bZ],\rho_{\bZ})$\\
 \hline 
 \end{oldtabular}
\end{equation}
Second, we define a CAS program
\begin{equation}\label{eq:casp-program}
P_1 = \langle \Pi_1, \cB_{\cL_3}, \gamma_1 \rangle
\end{equation}
over integer lexicon $\cL_3$, where
\begin{itemize}
	\item $\Pi_1$ is the program
\begin{equation}\label{eq:casp}
\ba l
  \{switch\}.\\
  lightOn\ar\ switch, not\ am.\\
  \ar not\ lightOn. \\
  \{am\}.\\
  \ar not\ am, |x<12|.\\
  \ar am, |x \geq 12|.\\
\ar |x<0|.\\
\ar |x>23|.\\
\ea
\end{equation}
The set of irregular atoms of~$\Pi_1$ is~$\{|x<12|,|x\geq12|,|x<0|,|x>23|\}$. The remaining atoms form the regular set.
The first four lines of program $\Pi_1$ are identical to these of logic program \eqref{eq:acp}. The last four lines of the program state:
\begin{itemize}[noitemsep]
\item It must be $am$ when $x<12$, where $x$ is understood as the hours.
\item It is impossible for it to be $am$ when $x \geq 12$.
\item Variable $x$ must be nonnegative.
\item Variable $x$ must be less than or equal to $23$.
\end{itemize}
\bigskip

	\item $\cB_{\cL_3}$ is the set of all integer linear constraints over integer lexicon ${\cL_3}$, which obviously includes constraints $\{x<12,x\geq 12,x<0,x>23\}$, 
	
	\medskip
	\item $\gamma_1(a) = \begin{cases}
		\mbox{constraint $x<12$  over integer lexicon $\cL_3$}  &\mbox{if } a = |x<12| \\
		\mbox{constraint $x\geq 12$ over integer lexicon $\cL_3$}  &\mbox{if } a = |x \geq 12|.\\
		\mbox{constraint $x<0$  over integer lexicon $\cL_3$}  &\mbox{if } a = |x<0| \\
		\mbox{constraint $x> 23$ over integer lexicon $\cL_3$}  &\mbox{if } a = |x > 23|.
	\end{cases}$ 
\end{itemize}

\bigskip
\noindent
Consider the set 
\begin{equation}\label{eq:caspAnswerSet}
	\{switch, \ lightOn, |x \geq 12|\}
\end{equation}
over $At(\Pi_1)$. This set  is the only input answer set of $\Pi_1$ relative to its irregular atoms. 
Also, the 
 integer linear constraint satisfaction problem
 with constraints $$
 \ba{c}
 \{\gamma_1(|x \geq 12|), 
 \neg \gamma_1(|x < 12|),
 \neg \gamma_1(|x < 0|),
 \neg \gamma_1(|x > 23|)
 \}\\
 =\\
 \{x\geq12, 
 \neg x<12,
 \neg x<0,
 \neg x>23 
 \}
 \ea
 $$ has a solution. There are 12 valuations ${v_1}\dots v_{12}$ relative to integer lexicon $\cL_3$ for $x$, which satisfy this GCSP: $x^{v_1}=12,\dots,x^{v_{12}}=23$. 
It follows that
set~\eqref{eq:caspAnswerSet} is an answer set of $P_1$.
Pair 	$\langle \{switch, \ lightOn, |x \geq 12|\},\nu_1\rangle$ is one of the twelve extended answer sets of $P_1$.
\end{example}

\subsubsection{On Grounding}
In practice, CASP languages 
similarly to ASP languages, allow for non-constraint variables. 
\citeauthor{geb09}~(\citeyear{geb09}) present a program written in the language supported by CASP solver {\clingcon} for the so called {\em bucket} problem. We list a rule from that program to illustrate the notion of non-constraint variables\footnote{A rule is taken form the site documenting system {\clingcon}:  \url{http://www.cs.uni-potsdam.de/clingcon/examples/bucket_torsten.lp} .}:
\begin{verbatim}
:- volume(B,T+1) $!= volume(B,T) $+ A, pour(B,T,A), 
   bucket(B), time(T), amount(A).
\end{verbatim}
Non-constraint variables of these rules are $B$, $T$, and $A$.  
The  sign $\$$  marks the irregular atoms. 
Rules of the kind are interpreted 
as  shorthands for the set of rules without non-constraint variables, called {\em ground rules}. Ground rules are obtained by replacing
every non-constraint variable in the rules by suitable
terms not containing non-constraint variables.
For instance, if suitable terms for 
non-constraint variable $B$ in the sample rule above range over a single value $a$;
suitable terms for
$T$ range over two values $0$ and $1$; and suitable terms for 
$A$ range over two values $1$ and $2$, then the {\clingcon} rule above is instantiated as follows:
\begin{verbatim}
:- volume(a,1) $!= volume(a,0) $+ 1, pour(a,0,1), 
   bucket(a), time(0), amount(1).
:- volume(a,2) $!= volume(a,1) $+ 1, pour(a,1,1), 
   bucket(a), time(1), amount(1).
:- volume(a,1) $!= volume(a,0) $+ 2, pour(a,0,2), 
   bucket(a), time(0), amount(2).
:- volume(a,2) $!= volume(a,1) $+ 2, pour(a,1,2), 
   bucket(a), time(1), amount(2).
\end{verbatim}
To help to map these rules into the syntax used earlier, we rewrite the first one using notation of this paper:
$$
\ar |volume(a,1) \neq volume(a,0) + 1|, pour(a,0,1), 
bucket(a), time(0), amount(1).
$$
Expressions $volume(a,0)$, $volume(a,1)$, and $volume(a,2)$ stand for  constraint variables in the signature of integer lexicon of this program. System \clingcon supports integer linear constraints, where a constraint atom $|volume(a,1) \neq volume(a,0) + 1|$ is naturally mapped into integer linear constraint $volume(a,1) \neq volume(a,0) + 1$.

The process of replacing non-ground rules by their ground
counterparts is called {\em grounding} and is well understood in
ASP~\cite{geb07b,cal08}. 
The availability of non-constraint variables makes modeling in CASP an attractive and easy process. In fact, this feature distinguishes greatly CASP and SMT paradigms. An SMT-LIB language~\cite{smt15} is a standard language for interfacing major SMT solvers. This language  
 does not provide a user with convenient modeling capabilities. There is an underlying assumption that code in SMT-LIB is to be generated by special-purpose programs. 
   
Since grounding is a stand-alone process in CASP and a non-ground program is viewed as a shorthand for ground programs the focus of this paper is on the later.

\subsubsection{Input Completion} Similar to how completion was defined in Section~\ref{sec:intro}, we now define an input completion which is relative to an (input) vocabulary.
\begin{definition}
For a program $\Pi$ over vocabulary $\sigma$,
the {\em input-completion} of $\Pi$ relative to vocabulary~\hbox{$\iota\subseteq\sigma$} so that $\hd(\Pi)\cap\iota=\emptyset$, denoted by~$IComp(\Pi,\iota)$, is defined as the set of formulas in propositional logic that consists of the  rules~\eqref{e:prop-formula} in $\Pi$ and the implications~\eqref{eq:comp2}
for all atoms~$a$ occurring in $\sigma\setminus\iota$.
\end{definition}

\begin{example}\label{ex:input-comp} Here we illustrate the concept of input completion. Consider program $\Pi_1$ from Example~\ref{ex:casp}.
	Its input completion relative to a vocabulary consisting of its irregular atoms~$\{|x<12|,|x\geq12|,|x<0|,|x>23|\}$ consists of formulas in~\eqref{eq:acp-comp2} and formulas in
\begin{equation} \label{eq:acp-icomp1}
\ba l
  \neg am\wedge |x<12| \rar \bot\\
  am\wedge |x \geq 12| \rar \bot.\\
  |x<0|\rar \bot.\\
  |x>23|\rar \bot.\\
\ea
\end{equation}
It is easy to see that $IComp(\Pi_1,\{|x<12|,|x \geq 12|,|x<0|,|x>23|\})$ is equivalent to the  union of~\eqref{eq:acp-comp} and~\eqref{eq:acp-icomp1}.
The set  $\{switch,lightOn,|x \geq 12|\}$ is the only model of this input completion. Note that this model coincides with the input answer set of $\Pi_1$ relative to the set of its irregular atoms.
\end{example}

The observation that we made last in the preceding example is an instance of the general fact captured by the following theorem. 
\begin{theorem}
\label{thm:input-ans-set-iff-sat-icomp}
 For a tight program $\Pi$ over vocabulary $\sigma$ and vocabulary $\iota\subseteq\sigma$ so that $\hd(\Pi)\cap \iota=\emptyset$, a set~$X$ of atoms from~$\sigma$ is an input answer set 
of~$\Pi$ relative to~$\iota$ if and only if $X$ satisfies the program's input-completion $IComp(\Pi,\iota)$.
\end{theorem}
Furthermore, for any program any of its input answer sets  is also a model of its input-completion.
\begin{theorem}
\label{thm:input-ans-set-iff-sat-icomp2}
 For a program $\Pi$ over vocabulary $\sigma$ and vocabulary $\iota\subseteq\sigma$ so that $\hd(\Pi)\cap \iota=\emptyset$, if a set~$X$ of atoms from~$\sigma$ is an input answer set 
of~$\Pi$ relative to~$\iota$ then $X$ satisfies the program's input-completion $IComp(\Pi,\iota)$.
\end{theorem}

To prove these theorems it is useful to state the following lemma.
\begin{lemma}\label{lem:comp}
 For a program $\Pi$ over vocabulary $\sigma$ and vocabulary $\iota\subseteq\sigma$ so that $\hd(\Pi)\cap \iota=\emptyset$, a set~$X$ of atoms over~$\sigma$ is a model of formula $IComp(\Pi,\iota)$ if and only if it is a model of 
 $Comp(\Pi \cup (X\cap\iota)).$
\end{lemma}
\begin{proof}
Since we  identify rules~\eqref{e:rulea:b} in $\Pi$  with respective implications~$B\rightarrow a$  we may write symbol~$\Pi$ to denote not only a set of  rules but also a set of respective implications.

We can write $IComp(\Pi,\iota)$ as the union of
\beq
\Pi,
\eeq{eq:prog} 
implications
\beq
\{a \rar \bigvee_{B\in Bodies(\Pi,a)} B ~|~  a \in \hd(\Pi)\},
\eeq{eq:head-imply-body}
and
\beq
\{a \rar \bot ~|~  a \not \in \hd(\Pi) \hbox{ and } a\in \sigma\setminus\iota\}.
\eeq{eq:head-imply-body2}
Note that expression~\eqref{eq:head-imply-body2} can be written as 
\beq
\{a \rar \bot ~|~  a \not \in \hd(\Pi) \hbox{ and } a\not \in \iota \hbox{ and } a\in \sigma\}.
\eeq{eq:head-imply-body3}

Similarly, we can write $Comp(\Pi \cup (X\cap\iota))$ as a union of \eqref{eq:prog},
\beq
(X\cap\iota) 
\eeq{eq:ans-diff-prog}
\beq
\{a \rar \bigvee_{B\in Bodies(\Pi\cup (X\cap\iota),a)} B ~|~  a \in \hd(\Pi)\}
\eeq{eq:head-imply-body-union}
\beq
\{a \rar \bot ~|~ a \not \in \hd(\Pi) \hbox{ and } a  \not \in (X\cap\iota) \hbox{ and } a\in \sigma\}
\eeq{eq:head-imply-false}

Left-to-right:
Assume $X \models IComp(\Pi,\iota)$.
It consists of~\eqref{eq:prog},~\eqref{eq:head-imply-body}, and~\eqref{eq:head-imply-body3} by construction.
 Trivially, $X$ satisfies~\eqref{eq:ans-diff-prog}. 
Thus we are left to show that $X$ satisfies
  \eqref{eq:head-imply-body-union} and~\eqref{eq:head-imply-false}. 
  Note that  \eqref{eq:head-imply-body} and
\eqref{eq:head-imply-body-union}  coincide since $Bodies(\Pi,a)=Bodies(\Pi \cup (X\cap\iota)),a)$ for all atoms $a \in \hd(\Pi)$ as $\hd(\Pi)\cap\iota=\emptyset$. 
Consequently,~$X$ satisfies  \eqref{eq:head-imply-body-union}.
Observe set~\eqref{eq:head-imply-false} can be written as the union of set~\eqref{eq:head-imply-body3}
and set 
\beq
\{a \rar \bot ~|~ a \not \in \hd(\Pi) \hbox{ and } a  \in (\iota\setminus X) \hbox{ and } a\in \sigma\}.
\eeq{eq:head-imply-false2}
Trivially, $X$ satisfies~\eqref{eq:head-imply-false2} as any atom that satisfies condition $a  \in (\iota\setminus X)$ is such that $a\not\in X$. 
It also satisfies~\eqref{eq:head-imply-body3}. Consequently, $X$ satisfies~\eqref{eq:head-imply-false}. 
 
Right-to-left:
Assume $X$ is a model of 
 $Comp(\Pi \cup (X\cap\iota)).$
 Set~$X$ satisfies~\eqref{eq:prog}, \eqref{eq:ans-diff-prog}, \eqref{eq:head-imply-body-union}, and \eqref{eq:head-imply-false}.
 Since $X$ satisfies \eqref{eq:head-imply-body-union}, it  satisfies \eqref{eq:head-imply-body} (as they coincide as argued above). 
  Since $X$ satisfies~\eqref{eq:head-imply-false}, $X$ also satisfies~\eqref{eq:head-imply-body3} (Recall  set~\eqref{eq:head-imply-false} can be written as the union of set~\eqref{eq:head-imply-body3}
  and set~\eqref{eq:head-imply-false2}). 
Consequently,  $X \models IComp(\Pi,\iota)$. 
\end{proof}

\begin{proof} [Proof of Theorem~\ref{thm:input-ans-set-iff-sat-icomp}]
We are given that $\Pi$ is tight.
 Since $X\cap\iota$ only consists of facts, it follows that~$\Pi \cup (X\cap\iota)$ is tight also.

Left-to-right: Assume $X$ is an input answer set of a program $\Pi$ relative to $\iota$. 
By Definition~\ref{def:input-answer-set},~$X$ is  an answer set of $\Pi\cup (X\cap\iota)$.  
It is well known that an answer set of a program is also a model of its completion. Thus,~$X$ is a model of~$Comp(\Pi \cup (X\cap\iota))$.
By Lemma~\ref{lem:comp},~$X$ is a model of $IComp(\Pi,\iota)$.
 
Right-to-left: Assume $X \models IComp(\Pi,\iota)$. 
By Lemma~\ref{lem:comp}, $X$ is a model of completion
 \hbox{$Comp(\Pi \cup (X\cap\iota))$}. 
 By results by Fages~\shortcite{fag94} and  Erdem and Lifshitz~\shortcite{erd01}, $X$ is the answer set of the program $\Pi \cup (X\cap\iota)$. By Definition~\ref{def:input-answer-set}, $X$ is an input answer set of~$\Pi$ relative to $\iota$.
\end{proof}
Direction Left-to-right of the proof of Theorem~\ref{thm:input-ans-set-iff-sat-icomp} serves as a proof for Theorem~\ref{thm:input-ans-set-iff-sat-icomp2} also.

\subsection{Constraint Formula}
Just as we defined constraint answer set programs, we can define constraint  formulas.

For a propositional formula $F$, by $\At(F)$ we denote the set of atoms occurring in it.
\begin{definition}
A {\em constraint  formula}  
over the vocabulary $\sigma=\sigma_r
\cup\sigma_i$  is a triple $\langle F,\cB,\gamma\rangle$, where 
\begin{itemize}
\item 	$F$ is a propositional formula 
over the vocabulary~$\sigma$,
\item $\cB$ is a set of constraints over some lexicon~$\cL$,
and
\item $\gamma$ is an injective function from the set~$\sigma_i$ of 
irregular atoms to the set $\cB$ of constraints.
\end{itemize}

For a constraint formula $\F=\langle F,\cB,\gamma\rangle$ over the vocabulary $\sigma=
\sigma_r\cup\sigma_i$ such that $\cL$ is the lexicon of the constraints in $\cB$, a set $X\subseteq\At(F)$ 
is a \emph{model} of $\F$ if
\begin{enumerate}
\item  $X$ is a model of $F$, and 
\item the following GCSP over $\cL$ has a solution
\[\{\gamma(a) | a\in X\cap\sigma_i\}\cup
\{\neg{\gamma(a)} | a\in  (\At(F)\cap \sigma_i)\setminus X\}.\]%
\end{enumerate}
\end{definition}

\begin{example}\label{ex:constraint-formula} 
Similar to the CAS program $P_1$ from Example~\ref{ex:casp}, we can define a constraint formula 
\[\cF_1 = \langle IComp(\Pi_1,\{|x<12|,|x \geq 12|,|x<0|,|x>23|\}),\cB_{\cL_3},\gamma_1\rangle\]
relative to integer lexicon $\cL_3$. We understand $\Pi_1$, $\cB_{\cL_3}$, and $\gamma_1$ as in Example~\ref{ex:casp}. Example~\ref{ex:input-comp}
illustrates how input completion $IComp(\Pi_1,\{|x<12|,|x \geq 12|,|x<0|,|x>23|\})$ is formed.
The set 
\[\{switch,lightOn,|x \geq 12|\}\]
is the only model of $\cF_1$. 
\end{example}

Following theorem captures a relation between CAS programs and constraint formulas. This theorem is an immediate consequence of  Theorem \ref{thm:input-ans-set-iff-sat-icomp}.

\begin{theorem}\label{thm:casptight}
For a CAS program $P=\langle \Pi,\cB,\gamma\rangle$ over the vocabulary $\sigma=
\sigma_r\cup\sigma_i$ 
and a set $X$ of atoms over $\sigma$, when  $\Pi$ is tight, 
$X$ is an answer set of $P$ if and only if $X$ is a model of constraint formula  $\langle IComp(\Pi,\sigma_i),\cB,\gamma\rangle$ over $\sigma=
\sigma_r\cup\sigma_i$. 
\end{theorem}

We note that \exref{ex:casp} and \exref{ex:constraint-formula} demonstrate this property. 
In the future we will abuse the term "tight". We will refer to CAS program  $P=\langle \Pi,\cB,\gamma\rangle$ as {\em tight} when its first member $\Pi$ has this property.


\section{Satisfiability Modulo Theories}\label{sec:smt}

First, in this section we  introduce the notion of a ``theory'' in satisfiability modulo theories~(SMT)~\cite{BarTin-14}.
Second, we present the definition of a ``restriction formula'' and state the conditions under which such formulas are satisfied by a given interpretation. 
These formulas are syntactically restricted classical ground predicate logic formulas. To be precise, a restriction formula  corresponds to a conjunction of ground literals.
The presented notions of interpretation and satisfaction are usual, but are stated in terms convenient for our purposes.
 
\begin{definition}\label{def:theory}
An {\em interpretation}~$I$ for a signature~$\Sigma$, or \hbox{\em $\Sigma$-interpretation}, 
is a tuple $(D,\nu,\rho,\phi)$ where 
\begin{itemize}
\item $D$ is a domain,
\item  $\nu$ is a $[\var{\Sigma},D]$ valuation,
\item $\rho$ is a $[\rels{\Sigma},D]$ r-denotation, and
\item $\phi$ is a $[\fsym{\Sigma},D]$ f-denotation.
\end{itemize}
For a signature $\Sigma$, a {\em $\Sigma$-theory} is a set of interpretations over $\Sigma$.
\end{definition}

For signatures that contain no f-symbols, we omit the reference to the last element of the interpretation tuple.
 For instance, for signature~$\Sigma_1$ from Table~\ref{tbl:cspDefinitions}, consider the following sample interpretations:
 \begin{equation}\label{eq:interpretations}
  \begin{oldtabular}{ll}
  \hline 
   $\cI_1$ ~~~~~~~~~~ & $(D_1,\nu_1,\rho_1)$ \\
   $\cI_2$ & $(D_1,\nu_2,\rho_1)$ \\
  \hline 
  \end{oldtabular}
 \end{equation}
Any subset of interpretations $\{\cI_1,\cI_2\}$ exemplifies a unique $\Sigma_1$-theory.

As mentioned earlier, in  literature on predicate logic and SMT, the terms  {\em object constant} and {\em function symbol of arity~$0$} are commonly used to refer to elements in the signature that we call variables.
Here we use the terms that stem from definitions related to constraint satisfaction processing to facilitate uncovering the precise link between CASP-like formalisms and SMT-like formalisms. It is typical for predicate logic signatures to contain propositional symbols (predicate symbols of arity~$0$).
It is easy to extend the notion of signature introduced here to allow propositional symbols. Yet it will complicate the presentation, which is the reason we avoid this extension.

A {\em restriction formula} over a signature $\Sigma$
is a finite set of constraint literals over a \hbox{c-vocabulary} $[\Sigma,\emptyset]$.
A sample restriction formula over $\Sigma_1$ follows
\beq
\{\neg E(s), \neg Q(r,s)\}.
\eeq{eq:tformula}

We now introduce the semantics of restriction formulas.
Let~$I=(D,\nu,\rho,\phi)$ be  a \hbox{$\Sigma$-interpretation}.
To each term~$\tau$ over a c-vocabulary $[\Sigma,\emptyset]$,
$I$ assigns a value $\tau^{\nu,\phi}$ that we denote by~$\tau^I$.
We say that interpretation~$I$ {\em satisfies} restriction formula~$\Phi$ over~$\Sigma$, when~$\nu$ satisfies every constraint literal in~$\Phi$ with respect to~$\rho$ and~$\phi$. 
For instance, interpretation $\cI_2$ satisfies restriction formula~\eqref{eq:tformula}, while $\cI_1$ does not satisfy~\eqref{eq:tformula}.

We say that a restriction formula $\Phi$ over signature $\Sigma$ is {\em satisfiable} in a $\Sigma$-theory $T$, or is $T$-satisfiable, when there is an element of the set $T$ that satisfies $\Phi$. 
For example, restriction formula~\eqref{eq:tformula} is satisfiable in  any $\Sigma_1$-theory that contains  interpretation $\cI_2$. 
On the other hand,  restriction formula~\eqref{eq:tformula} is not satisfiable in  $\Sigma_1$-theory $\{\cI_1\}$.

To conclude the introduction to the concept of $\Sigma$-theory:  from the semantical perspective, it is a collection of $\Sigma$-interpretations;
from the syntactic perspective, the theory is
 the collection of restriction formulas satisfied by these models.

\subsection{Uniform Theories and Link to Generalized Constraint Satisfaction Processing}
The presented definition of a theory (Definition~\ref{def:theory}) places no restrictions on the domains, r-denotations, or f-denotations being identical across the interpretations defining a theory. In practice, such restrictions are very common in SMT. 
We  now define ``uniform'' theories that follow these typical restrictions. We will then show how restriction formulas interpreted over uniform theories can practically be seen as   GCSPs.

\begin{definition}\label{ref:defu}
	For a signature $\Sigma$, we call a $\Sigma$-theory $T$ {\em uniform} over lexicon~$\cL=([\Sigma,D],\rho,\phi)$  when
	\begin{enumerate} 
		\item\label{l:defu1} all interpretations in $T$ are of the form $(D,\nu,\rho,\phi)$ (note how valuation $\nu$ is the only not fixed element in the interpretations), and
		\item\label{l:defu2} for every possible $[\var{\Sigma},D]$ valuation $\nu$, there is an interpretation 
		$(D,\nu,\rho,\phi)$ in $T$.
	\end{enumerate}	
\end{definition}

\begin{example}\label{ex:uniform}
To illustrate a concept of a uniform theory,
a table below defines  sample domain~$D_2$, valuations~$\nu_3$ and~$\nu_4$, and r-denotation~$\rho_4$.
\begin{center}
 \begin{oldtabular}{ll}
 \hline 
  $D_2$~~~~~~~~~~  & $\{1,2\}$\\
  $\nu_3$ &  $[\var{\Sigma_1},D_2]$ valuation, where $s^{\nu_3}=1$ and $r^{\nu_3}=2$\\
$\nu_4$ & $[\var{\Sigma_1},D_2]$ valuation, where $s^{\nu_4}=2$ and $r^{\nu_4}=2$\\
$\rho_4$ & $[\rels{\Sigma_1},D_2]$ r-denotation, where 
$
E^{\rho_4}=\{\langle 2 \rangle\},
Q^{\rho_4}=\{\langle 1,1 \rangle,\langle 2,2 \rangle\}$\\
\hline 
 \end{oldtabular}
 \end{center}
We also note that valuations $\nu_1$ and $\nu_2$ from Table~\ref{tbl:cspDefinitions} can be seen not only as $[\var{\Sigma_1},D_1]$ valuations, but also as $[\var{\Sigma_1},D_2]$ valuations.

The set
\beq
\{(D_2,\nu_1,\rho_4),(D_2,\nu_2,\rho_4),(D_2,\nu_3,\rho_4),(D_2,\nu_4,\rho_4)\}
\eeq{eq:t1}
of $\Sigma_1$ interpretations is an example of a uniform theory over lexicon $([\Sigma_1,D_2],\rho_4)$. 
The set
$$
\{(D_2,\nu_1,\rho_4),(D_2,\nu_2,\rho_4),(D_2,\nu_3,\rho_4),(D_1,\nu_4,\rho_4)\}
$$
of $\Sigma_1$ interpretations is an example of a non-uniform theory. Indeed, the condition~\ref{l:defu1} of Definition~\ref{ref:defu}    does not hold for this theory: the last interpretation refers to a different domain than the others. 
Also, recall interpretations $\cI_1$ and $\cI_2$ given in~\eqref{eq:interpretations}. 
Neither $\Sigma_1$-theory~$\{\cI_1\}$ nor $\{\cI_1,\cI_2\}$ 
 is uniform over lexicon~$([\Sigma_1,D_1],\rho_1)$. In this case, the condition~\ref{l:defu2} of Definition~\ref{ref:defu} does not hold. 
\end{example}

It is easy to see that for uniform theories we can identify their interpretations of the form $(D,\nu,\rho,\phi)$ with their second element -- valuation $\nu$. Indeed, the other three elements are fixed by the lexicon over which the uniform theory is defined.
In the following, we will sometimes use this convention. For example, we may refer to 
interpretation $(D_2,\nu_1,\rho_4)$ of uniform theory~\eqref{eq:t1} as $\nu_1$.

For uniform $\Sigma$-theory $T$ over lexicon $([\Sigma,D],\rho,\phi)$, 
we can extend the syntax of restriction formulas by saying that a {\em restriction formula} is defined over c-vocabulary $[\Sigma,D]$ as a finite set  of constraint literals over $[\Sigma,D]$ (earlier we considered constraint literals over~$[\Sigma,\emptyset]$). The earlier definition of semantics is still applicable. 

We now present a theorem that makes the connection between GCSPs over a lexicon~$\cL$ and restriction formulas interpreted using the uniform theory $T$ over the same lexicon  apparent. As a result, the question whether a given GCSP over $\cL$ has a solution  translates into the question whether the set of constraint literals of a GCSP forming a restriction formula is $T$-satisfiable. Furthermore, any solution to such GCSP is also an interpretation in $T$ that satisfies the respective restriction formula, and the other way around. 

\begin{theorem}\label{thm:smtcas}
	For a lexicon $\cL=([\Sigma,D],\rho,\phi)$, a set $\Phi$ of constraint literals over c-vocabulary~$[\Sigma,D]$, and a uniform $\Sigma$-theory $T$ over lexicon $\cL$ the following holds
	\begin{enumerate}
		\item\label{l:t3:1} for any $[\var{\Sigma},D]$ valuation $\nu$, there is an interpretation $\nu$ in $T$, 
		\item\label{l:t3:2} $[\var{\Sigma},D]$ valuation $\nu$ is a solution to GCSP $\Phi$ over lexicon $\cL$ if and only if  interpretation $\nu$ in $T$ satisfies restriction formula $\Phi$,
		\item\label{l:t3:3} GCSP $\Phi$ over lexicon $\cL$ has a solution if and only if  restriction formula $\Phi$ is $T$-satisfiable. 
	\end{enumerate}
\end{theorem}
 \begin{proof}
 	Statement~\ref{l:t3:1} trivially follows from the condition~\ref{l:defu2} of the definition of uniform theories.
 	
 	Proof of Statement~\ref{l:t3:2}. 
 	By Statement~\ref{l:t3:1}, interpretation $\nu$ is in $T$.
 	By definition of a solution to GCSP,  $[\var{\Sigma},D]$ valuation $\nu$ is a solution to GCSP $\Phi$ over lexicon $\cL$ if and only if $\nu$ is a solution to every constraint in $\Phi$. In other words, $\nu$ satisfies every constraint literal in $\Phi$ with respect to~$\rho$ and~$\phi$. By definition of interpretations satisfying formulas, the previous statement holds if and only if interpretation $\nu$ in $T$ satisfies $\Phi$.
 	
 	Proof of Statement~\ref{l:t3:3}. GCSP $\Phi$ over lexicon $\cL$ has a solution if and only if 
 	there is a  $[\var{\Sigma},D]$ valuation $\nu$ that is a solution to GCSP $\Phi$ over lexicon $\cL$.
 	By statements~\ref{l:t3:1} and~\ref{l:t3:2}, the previous statement holds if and only if there is an interpretation $\nu$ in $T$ that satisfies $\Phi$ and consequently  $\Phi$ is $T$-satisfiable. 
 \end{proof}


Such commonly used  theories in SMT  as  
{\em linear real arithmetic},  {\em linear integer arithmetic}, and {\em integer difference logic} are uniform.\footnote{For instance, the SMT solver {\sc cvc4}~\cite{cvc4} supports all three mentioned arithmetics.} To be more precise,
recall notions of numeric, integer lexicons, and (integer) linear constraints
  presented in Section~\ref{sec:licon}. Given these notions:
\begin{itemize}
\item Linear real arithmetic is an example of a 
  uniform theory over a numeric lexicon. This arithmetic poses syntactic conditions on restriction formulas that it interprets. Namely, literals in these restriction formulas must correspond to linear constraints.
  \item  Similarly, linear integer arithmetic and  integer difference logic
  are
   examples of   uniform theories over integer lexicons. 
   Literals in restriction formulas in these arithmetics must correspond to integer linear constraints. Furthermore,
   the difference logic is a special case of linear integer arithmetic posing yet additional syntactic restrictions~\cite{nie05}.
\end{itemize} 
From Theorem~\ref{thm:smtcas}, it follows that  restriction formulas in linear real arithmetic can be seen as  linear constraint satisfaction problems (and the other way around). 
Similar relation holds between restriction formulas in linear integer arithmetic and integer linear constraint satisfaction problems.

\section{SMT Formulas and ASPT Programs}\label{sec:smtformulas}
First, this section introduces 
 SMT formulas that merge the concepts of propositional formulas and $\Sigma$-theories. Second, it introduces
 ASPT  programs that merge the concepts of logic programs and $\Sigma$-theories.
 It turns out that if considered $\Sigma$-theories  are uniform then formalisms of constraint formulas and SMT formulas coincide. A similar relation holds of constraint answer set programs and ASPT programs. This link is discussed in the next section.

As in Section~\ref{sec:casp2}, we understand
 $\sigma_r$ and $\sigma_i$ as two disjoint vocabularies and refer to their elements as regular and irregular.

\begin{definition}
An {\em SMT formula} $P$ over vocabulary $\sigma=\sigma_r\cup\sigma_i$ is a triple
$\langle  F, T, \mu\rangle$, where
\begin{itemize}
\item $F$ is a propositional formula  over $\sigma$,
\item $T$ is a $\Sigma$-theory, and
\item $\mu$ is an injective function from irregular atoms~$\sigma_i$ to constraint literals over c-vocabulary $[\Sigma,\emptyset]$.
\end{itemize}
For an SMT formula $\langle  F, T, \mu\rangle$ over  $\sigma$, 
a set $X\subseteq\At(F)$ is its
{\em model} if 
\begin{enumerate}
\item $X$ is a model of $F$, and 
\item the following restriction formula 
\[\{\mu(a) | a\in X\cap\sigma_i\}\cup
\{\neg{\mu(a)} | a\in  (\At(F)\cap \sigma_i)\setminus X\}.\]
is satisfiable in $\Sigma$-theory $T$.
\end{enumerate}
\end{definition}
In the literature on SMT, a more sophisticated syntax than SMT formulas provide is typically discussed. 
Yet, SMT solvers often rely on the so-called propositional abstractions of predicate logic formulas~\citep[Section 26.2.1.4]{BarretSST08} or~\citep[Section 1.1]{BarTin-14}, which, in their most commonly used case, coincide with 
SMT formulas  discussed here.

We now define the concept of logic programs modulo theories.
\begin{definition}
A {\em logic program modulo theories} (or {\em ASPT program}) 
$P$ over vocabulary $\sigma=\sigma_r\cup\sigma_i$ is a triple
$\langle  \Pi, T, \mu\rangle$, where
\begin{itemize}
\item
$\Pi$ is a logic program over $\sigma$  such that $\hd(\Pi)\cap\sigma_i=\emptyset$, 
\item  $T$ is a $\Sigma$-theory,
and
\item $\mu$ is an injective function from irregular atoms $\sigma_i$ to constraint literals over c-vocabulary $[\Sigma,\emptyset]$.
\end{itemize}
For an ASPT program $\langle  \Pi, T, \mu\rangle$ over  $\sigma$, 
a set $X\subseteq\At(\Pi)$ is its
{\em answer set} if 
\begin{enumerate}
\item $X$ is an input answer set of $\Pi$ relative to $\sigma_i$, and 
\item the following restriction formula 
\[\{\mu(a) | a\in X\cap\sigma_i\}\cup
\{\neg{\mu(a)} | a\in  (\At(\Pi)\cap \sigma_i)\setminus X\}.\]
is satisfiable in $\Sigma$-theory $T$.
\end{enumerate}
\end{definition}

In the case of uniform theories, we can extend the definition of SMT formula given a constraint $\Sigma$-theory $T$ over lexicon $([\Sigma,D],\rho,\phi)$ as follows: an {\em SMT formula} $P$ over vocabulary $\sigma=\sigma_r\cup\sigma_i$ is a triple
$\langle  F, T, \mu\rangle$, where~$F$ is a propositional formula  over~$\sigma$,~$T$ is a $\Sigma$-theory, and~$\mu$ is an injective function from irregular atoms $\sigma_i$ to constraint literals over c-vocabulary $[\Sigma,D]$.
 Note how the only difference in this definition is that function~$\mu$ refers to domain $D$ of lexicon in place of an empty set.
The definition of ASPT program can be extended in the same style.
For the case of uniform theories, we will assume the definitions of SMT formulas
and ASPT programs as stated in this paragraph. 

\section{SMT Formulas versus Constraint Formulas and ASPT versus CAS Programs}\label{sec:relation}
The framework of uniform theories brings us to a straightforward relation between SMT formulas  over uniform theories and constraint formulas  and between CAS programs and ASPT programs.
We now formalize these statements.

Let $\cL$ denote a lexicon~$([\Sigma,D],\rho,\phi)$. 
By $\cB_\cL$ we denote the set of all constraints over~$\cL$. By  $T_\cL$ we denote the uniform $\Sigma$-theory over $\cL$. 
\begin{theorem}\label{thm:smt-cf}
For  a lexicon $\cL=([\Sigma,D],\rho,\phi)$,
 a vocabulary $\sigma=\sigma_r\cup\sigma_i$, and a set $X$ of atoms over~$\sigma$, 	set~$X$ is a model of  SMT formula $\langle F,T_\cL,\mu\rangle$ over $\sigma$
 	if and only if $X$ is a model of a constraint formula $\langle F,\cB_\cL,\mu\rangle$ over $\sigma$ (where $\mu$ is identified with the function from irregular atoms to constraints over $\cL$ in a trivial way.)
\end{theorem}
\begin{proof}
Let $X$ be a subset of $\At(F)$. 
	Set $X$ is a model of an SMT formula $\langle F,T_\cL,\mu\rangle$ over  $\sigma$
	if and only if 
	\begin{enumerate}
		\item $X$ is a model of $F$, and 
		
		\item\label{l:thm4:3} 
		the following restriction formula 
		\beq\{\mu(a) | a\in X\cap\sigma_i\}\cup
		\{\neg{\mu(a)} | a\in  (\At(F)\cap \sigma_i)\setminus X\}.
		\eeq{eq:thmp2}
		is satisfiable in $\Sigma$-theory $T_\cL$.		
	\end{enumerate}
	By the definition of a model of a constraint formula, to conclude the proof, we only have to illustrate that condition~\ref{l:thm4:3} holds if and only if 
	the GCSP~\eqref{eq:thmp2} over~$\cL$ 
	has a solution.
 By Statement~\ref{l:t3:3} of Theorem~\ref{thm:smtcas} we conclude that GSCP~\eqref{eq:thmp2} has a solution if and only if restriction formula~\eqref{eq:thmp2}
	is satisfiable in $\Sigma$-theory $T_\cL$ (which is the case if and only if condition~\ref{l:thm4:3} holds). 
\end{proof}

This theorem illustrates that for uniform theories the languages of SMT formulas and constraint formulas practically coincide. In other words,  constraint formulas is a special case of SMT formulas that are defined over uniform theories. It is obvious that a similar relation between CAS and ASPT programs holds.
\begin{theorem}\label{thm:cas-aspt}
For  a lexicon $\cL=([\Sigma,D],\rho,\phi)$,
 a vocabulary $\sigma=\sigma_r\cup\sigma_i$, and a set $X$ of atoms over~$\sigma$, 	set~$X$ is an answer set of  ASPT program $\langle \Pi,T_\cL,\mu\rangle$ over $\sigma$
 	if and only if $X$ is an answer set of a CAS program $\langle \Pi,\cB_\cL,\mu\rangle$ over $\sigma$ (where $\mu$ is identified with the function from irregular atoms to constraints over $\cL$ in a trivial way.) 	
\end{theorem}
The proof for this theorem 	follows the lines of proof of Theorem~\ref{thm:smt-cf}.

\begin{example}
Recall Examples~\ref{ex:casp} and~\ref{ex:constraint-formula}. 
Let $T_{\cL_3}$ denote the uniform theory over integer lexicon $\cL_3$ defined in~Example~\ref{ex:casp}. 

By Theorem~\ref{thm:cas-aspt}, CAS program $P_1=\langle \Pi_1,\cB_{\cL_3},\gamma_1 \rangle$ 
from Example~\ref{ex:casp}
is essentially the same entity as ASPT program  $\langle \Pi_1,T_{\cL_3},\gamma_1 \rangle$. These programs have the same answer sets.

	Similarly, by Theorem~\ref{thm:smt-cf},
constraint formula
\[\cF_1 = \langle IComp(\Pi_1,\{|x<12|,|x \geq 12|,|x<0|,|x>23|\}),\cB_{\cL_3},\gamma_1\rangle\]
from Example~\ref{ex:constraint-formula}
is essentially the same entity as SMT formula
$$\langle IComp(\Pi_1,\{|x<12|,|x \geq 12|,|x<0|,|x>23|\}),T_{\cL_3},\gamma_1 \rangle.$$
\end{example}


We call any SMT formula $\langle  F, T, \mu\rangle$ over $\sigma_r\cup\sigma_i$
\begin{itemize}
\item an {\em SMT(IL)~ formula}  
if $T$ is the uniform theory over an integer  lexicon
and $\mu$ maps irregular atoms $\sigma_i$ into integer linear constraints;
\item  an {\em SMT(DL) formula}  
if $T$ is the uniform theory over an integer lexicon and 
$\mu$ maps irregular atoms $\sigma_i$ into difference logic constraints;
\item an {\em SMT(L)~ formula}  
if $T$ is the uniform theory over a numeric  lexicon
and $\mu$ maps irregular atoms $\sigma_i$ into linear constraints.
\end{itemize} 
In the same style, we can define ASPT(IL),  ASPT(DL), and ASPT(L) programs.


From Theorem~\ref{thm:cas-aspt}, CAS programs of the form $\langle \Pi,\cB_\cL,\gamma\rangle$, where $\cL$ is a numeric lexicon  and $\cB_\cL$ is the set of all linear constraints over $\cL$, are essentially the same objects as  ASPT(L) programs.
Similarly, it follows that CAS programs of the form $\langle \Pi,\cB_\cL,\gamma\rangle$, where $\cL$ is an integer lexicon and $\cB_\cL$ is the set of all integer linear constraints over $\cL$, are essentially the same objects as  ASPT(IL) programs.


Obviously, Theorems~\ref{thm:casptight} and~\ref{thm:smt-cf} pave the way for using SMT systems that solve SMT(IL) and SMT(L) problems as is for solving  tight ASPT(IL) and ASPT(L) programs respectively. It is sufficient to compute the input completion of the program relative to irregular atoms.  \citeauthor{sus16b}~\shortcite{sus16b} utilized this method in implementing SMT-based solver for tight programs called {\sc ezsmt}. 
A similar observation has been used in work by Lee and Meng~\shortcite{lee13} and Janhunen~et al.~\shortcite{jan11}. Furthermore, Janhunen~et al. propose a translation of ASPT(DL) programs into SMT(DL) formulas. System~{\sc dingo} utilizes this translation by invoking SMT solver {\zThree} for finding models for ASPT(DL) programs.

\subsection{Outlook on Constraint Answer Set Solvers}\label{sec:caspil} 
We now relate various constraint answer set solvers by specifying which ASPT languages these solvers support.
We also provide a brief overview of a variety of solving techniques that they use.

Table \ref{tbl:casp-solvers} presents the landscape of current constraint answer set solvers. The star $*$ annotating language ASPT(IL) denotes that the solver supporting this language requires the specification of finite ranges for its variables (since finite-domain constraint solvers are used as underlying solving technology).
Symbol $\dagger$ that annotates the languages supported by system {\sc mingo}
represents the fact that this solver  supports lexicons that are not covered by our framework. Its variables within the same program can be of two different kind: either integer or real. 

\begin{table}[h]
\caption{Solvers Categorization}\label{tbl:casp-solvers}
\begin{tabular}{ll}
\hline \hline

Solver&Language\\
\hline
{\sc clingcon}~\cite{geb09}& ASPT(IL)$^*$\\
 {\sc mingo}~\cite{liu12}&ASPT(IL)$^\dagger$\\
 &ASPT(L)$^\dagger$\\
{\sc dingo}~\cite{jan11} &ASPT(DL)\\
{\sc ezcsp}~\cite{bal09a}&ASPT(IL)$^*$\\
&ASPT(IL)\\
&ASPT(L)\\
{\sc ezsmt}~\cite{sus16b}&ASPT(IL)\\
&ASPT(L)\\
\hline \hline
\end{tabular}
\end{table}

At a high-level abstraction, one may summarize the architectures of the {\sc clingcon} and {\sc ezcsp} solvers as 
{\em ASP-based solvers plus theory solver}. Given a CAS program $\langle \Pi,\cB,\gamma\rangle$,
both {\sc clingcon} and {\sc ezcsp} first use an answer set solver to partially compute an input answer set of~$\Pi$. Second, they contact a theory solver to verify whether respective constraint satisfaction problem has a solution. In case of {\sc clingcon},
  finite domain constraint solver {\sc gecode} is used as a theory solver.
  System {\sc ezcsp} uses constraint logic programming tools such as {\sc Bprolog}~\cite{b-prolog}, {\sc SICStus prolog}~\cite{sicstus-prolog}, and {\sc swi prolog}~\cite{swi-prolog}. These tools provide {\sc ezcsp} with the ability to work with three different kinds of constraints: finite-domain integer, integer linear, and linear constraints. 
  To process ASPT(IL) and ASPT(L) programs, the solver~{\sc mingo} translates these programs into mixed integer programming expressions and then uses the solver {\sc cplex}~\cite{cplex}
to solve these formulas.
To process ASPT(DL) programs {\sc dingo}  translates these programs into  SMT(DL) formulas and applies the SMT solver {\sc z3}~\cite{z3} to find their models.
The {\sc ezsmt} solver is only capable of processing tight programs. It computes clausified input completion for such CAS programs, encodes resulting SMT formula in 
SMT-LIB language~\cite{smt15} and is then able to invoke any SMT solver that supports SMT-LIB (for instance, SMT solvers {\sc z3}~\cite{z3} or {\sc cvc4}~\cite{cvc4}) to compute its models.

The diversity of solving approaches used in CASP paradigms suggests that solutions of the kind are available for SMT technology. Typical SMT architecture is in a style of systems {\sc clingcon} and {\sc ezcsp}. 
At a high-level abstraction, one may summarize common  architectures of SMT solvers as 
satisfiability-based solvers augmented with theory solvers.
Theory solvers are usually implemented within an SMT solver and are as such custom solutions.
The fact that {\sc clingon} and {\sc ezcsp} use tools available from the constraint programming community suggests that these tools could be of use in SMT community also. The solution exhibited by system {\sc mingo}, where mixed integer programming is used for solving ASPT(L) and ASPT(IL) programs, hints  that a similar strategy can be implemented for solving SMT(L) formulas. 
These ideas have  been explored by~\citeauthor{king14}~\shortcite{king14}.

\section{SMT Solvers for Nontight Programs via Level Rankings}\label{sec:nontight}

As mentioned in section on preliminaries, \citeauthor{nie08}~\shortcite{nie08} characterized answer sets of normal logic programs in terms of level rankings. He then developed a mapping from such programs to SMT(DL) formalism.
Mappings of the kind were  exploited in the design of solvers \dingo~\cite{jan11} and \mingo~\cite{liu12}.
  In this section, we devise  a similar mapping and show how it provides means for using SMT solvers for processing non-tight ASPT(IL) programs. 
  
  We start by  defining  level ranking relative to an input vocabulary and then lifting the result of Theorem~\ref{thm:ans-iff-lr} to the notion of input answer set.
 
 \begin{definition}\label{def:lrinput}
  A function $\lr: X\setminus\iota \rightarrow \bN$ is a {\em level ranking} of $X$ for $\Pi$ relative to vocabulary $\iota\subseteq\sigma$ so that $\hd(\Pi)\cap\iota=\emptyset$,  when  for each atom $a$ in $X\setminus\iota$ the following condition hold: 
  there is $B$ in $Bodies(\Pi,a)$ such that~$X$ satisfies $B$  and 
   	 for every $b \in B^+\setminus\iota$ it holds that $\text{\lr}(a) - 1 \geq \text{lr}(b)$.
  \end{definition}

 \begin{theorem}\label{thm:casp-ans-iff-lr}
	\label{thm:slr}
	For a program $\Pi$ over vocabulary $\sigma$, vocabulary $\iota\subseteq\sigma$ so that $\hd(\Pi)\cap\iota=\emptyset$, and a set~$X$ of atoms over $\sigma$ that is a model of input completion $IComp(\Pi,\iota)$, $X$ 
 is an input answer set of $\Pi$ relative to $\iota$ if and only if  there is a level ranking of $X$ for $\Pi$ relative to $\iota$.
\end{theorem}
\begin{proof} 
	Consider a set $X$ of atoms over $\sigma$ that is a model of  $IComp(\Pi,\iota)$.
	By Lemma~\ref{lem:comp} it follows that
	$X$ is a model of 
	\beq
	Comp(\Pi\cup (X\cap\iota)).
	\eeq{eq:compcomplex}
	
Left-to-right:	Assume set~$X$  is an input answer set of $\Pi$ relative to $\iota$. By the definition of an input answer set, $X$ is an answer set of the program 
\beq
\Pi\cup (X\cap\iota).
 \eeq{eq:prc}
By Theorem~\ref{thm:ans-iff-lr}, it follows that there is a  level ranking of $X$ for program~\eqref{eq:prc}.
Thus, by Definition~\ref{def:lr} there is a function \lr: $X \rightarrow \mathbb{N}$ such that for each $a \in X$, there is $B$ in $Bodies(\Pi\cup (X\cap\iota),a)$ such that~$X$ satisfies $B$  and 
for every $b \in B^+$ it holds that $\lr(a) - 1 \geq \lr(b)$.
It is easy to see that for each 
atom $a$ in $X\setminus\iota$, 
\beq
Bodies(\Pi\cup (X\cap\iota),a)=Bodies(\Pi,a).
\eeq{eq:bodieseq}
Thus, for each atom $a$ in $X\setminus\iota$
there is $B$ in $Bodies(\Pi,a)$ such that~$X$ satisfies $B$ and for every $b \in B^+$ it holds that $\lr(a) - 1 \geq \lr(b)$, and hence for every $b \in B^+\setminus\iota$ it holds that $\lr(a) - 1 \geq \lr(b)$. 
By Definition~\ref{def:lrinput},
 \lr (seen as a function from $X\setminus\iota$ to  $\bN$) is  a level ranking of $X$ for $\Pi$ relative to $\iota$.

Right-to-left: Assume there is a  level ranking \lr of $X$ for $\Pi$ relative to $\iota$.
 By Definition~\ref{def:lrinput},
 for each atom $a$ in $X\setminus\iota$
 there is $B$ in $Bodies(\Pi,a)$ such that~$X$ satisfies $B$ and  for every $b \in B^+\setminus\iota$ it holds that $\lr(a) - 1 \geq \lr(b)$.
 
 We now construct  function $\lr': X \rightarrow \bN$ based of \lr. 
 We then illustrate that $\lr'$ is a level ranking of $X$ for program~\eqref{eq:prc}. 
 
   For every atom $a$ in $X\cap\iota$, $\lr'(a)=0$.   
 For every atom $a$ in $X\setminus\iota$, $\lr'(a)=\lr(a)+1$.  
 It is easy to see that 
 $$
 (X\cap\iota)~\cup~(X\setminus\iota)=X
\hbox{ as well as }
 (X\cap\iota)~\cap~(X\setminus\iota)=\emptyset.
 $$
 Thus $\lr'$ is a function from $X$ to  $\mathbb{N}$.
 Consider two cases.
 
 Case 1. Atom $a$ in  
 $X\cap\iota$. Since $\hd(\Pi)\cap\iota=\emptyset$,  $a$ does not appear in any head in $\Pi$ (i.e., $a\not\in\hd(\Pi)$) and thus it only appears as a fact in
  $\Pi\cup (X\cap\iota)$. Trivially, level ranking condition for this atom is satisfied.
  
  Case 2. Atom $a$ in $X\setminus\iota$. It is easy to see that for such an atom equality~\eqref{eq:bodieseq}
  holds.
  We are given that  there is $B$ in $Bodies(\Pi,a)$ such that~$X$ satisfies $B$ and  for every $b \in B^+\setminus\iota$ it holds that $\lr(a) - 1 \geq \lr(b)$.
  From $\lr'$ construction, it follows that   
   for every $b \in B^+\setminus\iota$ it holds that $\lr'(a) - 1 \geq \lr'(b)$.
  On the other hand, for every atom 
  $b$ in $B^+$ that it is not in $B^+\setminus\iota$, such $b$ is in  
  $X\cap\iota$.
  By $\lr'$ construction, $\lr'(b)=0$ and
  $\lr'(a)=\lr(a)+1$. Since $\lr$ is a mapping to natural numbers we derive that $\lr(a)\geq 0$ and hence $\lr'(a)\geq 1$. Thus,  $\lr'(a)-1\geq \lr'(b)$ holds.
This concludes our illustration that for every atom $a\in X$ the level ranking condition holds given function~$\lr'$.   
 
By Theorem~\ref{thm:ans-iff-lr}  we conclude that $X$ is an answer set of program~\eqref{eq:prc}.
Consequently,
$X$ is an input answer set of  $\Pi$ relative to $\iota$.
\end{proof}

We now present a mapping from a logic program to SMT(IL) formula inspired by the mapping introduced by
\citeauthor{nie08}~\shortcite{nie08}.
The mapping $T(\Pi,\iota)$ consists of input-completion $IComp(\Pi,\iota)$
and ranking-formula $R(\Pi,\iota)$ defined below.
Sometimes, we refer to the set of formulas as a formula, where we understand such a formula as a conjunction of the members of its set.

In defining $R(\Pi,\iota)$, we will use the convention discussed earlier  so that vertical bars  mark the irregular atoms that  have intuitive mappings into respective integer linear constraints. For instance, expression $|lr_b-1\geq lr_a|$ introduces an irregular atom that is intuitively mapped into integer linear constraint
$lr_b-1\geq lr_a$, where $lr_a$ and $lr_b$ are variables over integers.
Expression $R(\Pi,\iota)$ is a conjunction of formulas constructed as follows:
for each atom $a\in\sigma\setminus\iota$  

	\beq
	a\rar \bigvee_{a\ar B\in \Pi \hbox{ and } B^+\setminus\iota\neq\emptyset}(B\wedge
	\bigwedge_{b\in B^+\setminus\iota}	
	|lr_a-1\geq lr_{b}|
	)\vee 
	\bigvee_{a\ar B\in \Pi \hbox{ and } B^+\setminus\iota=\emptyset}B		
	\eeq{eq:lrr}
	By $\sigma_i^{R(\Pi,\iota)}$ we denote the set of all irregular atoms of the form $|lr_a-1\geq lr_{b}|$ in 
	$R(\Pi,\iota)$.
	By $\Sigma^{R(\Pi,\iota)}$ we denote the set of all 
	variables over integers occurring within $\sigma_i^{R(\Pi,\iota)}$. 
	By~$\gamma^{R(\Pi,\iota)}$ we denote a mapping from 
	irregular atoms in $\sigma_i^{R(\Pi,\iota)}$ to integer linear constraints  in accordance to their intuitive meanings. By 
		$\cB_{R(\Pi,\iota)}$ we denote  the set of all integer linear constraints over integer lexicon defined over the signature $\Sigma^{R(\Pi,\iota)}$.
	
\begin{theorem}\label{thm:ians-set-iff-smt-lr-sat}
For a program $\Pi$ over vocabulary $\sigma$, vocabulary $\iota\subseteq\sigma$ so that $\hd(\Pi)\cap\iota=\emptyset$, and a set~$X$ of atoms over  $\sigma$,
$X$ is an input answer set 
of $\Pi$ relative to $\iota$ if and only if 
there is a model $X\cup X_i$ of a constraint formula  (or SMT(IL) formula) 
\beq \langle IComp(\Pi,\iota)\wedge R(\Pi,\iota),\cB_{R(\Pi,\iota)},\gamma^{R(\Pi,\iota)}\rangle\eeq{eq:icompform}
over $\sigma\cup\sigma_i^{R(\Pi,\iota)}$ so that $X_i\subseteq\sigma_i^{R(\Pi,\iota)}$.		
\end{theorem}
\begin{proof}
	
Left-to-right: 
Assume $X$ is an input answer set of $\Pi$ relative to $\iota$. By Theorem~\ref{thm:input-ans-set-iff-sat-icomp2}
it follows that $X$ satisfies $IComp(\Pi,\iota)$.  By Theorem~\ref{thm:casp-ans-iff-lr} there is a level ranking \lr of $X$ relative to~$\iota$. Consider such $\lr$.
By the level ranking definition, for each atom $a$ in $X\setminus\iota$ the following condition holds, there is $B$ in $Bodies(\Pi,a)$ such that~$X$ satisfies $B$  and 
   	 for every $b \in B^+\setminus\iota$ it holds that $\text{\lr}(a) - 1 \geq \text{lr}(b)$.
We now construct set $X_i$ of atoms over $\sigma_i^{R(\Pi,\iota)}$ as follows:
atom  $|lr_a-1\geq lr_{b}|$ from $\sigma_i^{R(\Pi,\iota)}$ is in $X_i$ if and only if condition  $\text{\lr}(a) - 1 \geq \text{lr}(b)$ holds.

It follows immediately from the construction of $X_i$ that the GCSP which corresponds to $X\cup X_i$ (in accordance with the condition~2 in the definition of a model of a constraint formula) has a solution. Indeed, consider a valuation $\nu$ that assigns values to variables of the form $lr_a$ in $\Sigma^{R(\Pi,\iota)}$ based on level ranking function $\lr$. In particular, $\lr_a^\nu=\lr(a)$.

We are now left to  illustrate that $X\cup X_i$ is a model of $IComp(\Pi,\iota)\wedge R(\Pi,\iota)$. 
Since~$X$ is a model of $IComp(\Pi,\iota)$, we are only left to show that
  $X\cup X_i$ is a model of $R(\Pi,\iota)$ or, in other words, that $X\cup X_i$ satisfies~\eqref{eq:lrr} for every atom $a\in\sigma\setminus\iota$.

Consider any atom $a\in\sigma\setminus\iota$.

Case 1. $a\not \in X$.
Obviously, $a\not\in X_i$ since $X_i\cap\sigma=\emptyset$. Then, $X\cup X_i$ trivially satisfies~\eqref{eq:lrr}.

Case 2.  $a \in X\setminus \iota$. 
Since \lr is a level ranking  of $X$ relative to $\iota$ it follows that 
 there is $B\in Bodies(\Pi,a)$ so that $X\models B$ and for every $b\in B^+\setminus\iota$ it holds that $\text{\lr}(a) - 1 \geq \text{lr}(b)$. Take such $B$.

Case 2.1 $B^+\setminus\iota=\emptyset$. Then, $X$  satisfies~\eqref{eq:lrr} due to 
$$	\bigvee_{a\ar B\in \Pi \hbox{ and } B^+\setminus\iota=\emptyset}B
	$$
	term in the right hand side of the implication~\eqref{eq:lrr} and the fact that $X$ satisfies considered~$B$.
Consequently, $X\cup X_i$  satisfies~\eqref{eq:lrr}.

Case 2.2 $B^+\setminus\iota\neq\emptyset$.
 From the construction of $X_i$,  it follows that for every 
$b\in B^+\setminus\iota$ there is an irregular atom of the form $|\text{\lr}(a) - 1 \geq \text{lr}(b)|$ in $X_i$. 
Then, $X\cup X_i$  satisfies~\eqref{eq:lrr} due to 
$$
\bigvee_{a\ar B\in \Pi \hbox{ and } B^+\setminus\iota\neq\emptyset}(B\wedge
\bigwedge_{b\in B^+\setminus\iota}	
|lr_a-1\geq lr_{b}|
)
$$
term in the right hand side of the implication~\eqref{eq:lrr}.

Right-to-left: Let $X\cup X_i$ be a model of constraint formula~\eqref{eq:icompform}.
It immediately follows that $X$ is a model of $IComp(\Pi,\iota)$.
We now show that one can construct  a level ranking \lr of $X$ for $\Pi$ relative to $\iota$ using $X_i$. Indeed, consider a GCSP that corresponds to $X\cup X_i$  (in accordance with the condition~2 in the definition of a model of a constraint formula). Since $X\cup X_i$ is a model of~\eqref{eq:icompform} that GCSP has a solution. We use such a solution $\nu$ to construct level ranking $\lr$ 
for 
each atom $a$ in $X\setminus\iota$
as follows: 
$$
\lr(a)=\begin{cases}
lr_a^\nu \hbox{~~ if }lr_a\in \Sigma^{R(\Pi,\iota)}\\
0 \hbox{~~~~~ otherwise }\\
\end{cases}
$$

To verify that $\lr$ is indeed a level ranking
of $X$ for $\Pi$ relative to $\iota$ using $X_i$ we have to illustrate that for every atom $a\in X\setminus\iota$
there is $B\in Bodies(\Pi,a)$ such that $X$ satisfies $B$ and for every $b\in B^+\setminus \iota$ it holds that $\lr(a)-1\geq\lr(b)$.
Consider any atom $a\in X\setminus\iota$.
We are given that  $X\cup X_i$ is a model of 
$R(\Pi,\iota)$. Thus, the right hand side of the implication~\eqref{eq:lrr} is satisfied for chosen atom $a$. It follows that there is $B\in Bodies(\Pi,a)$ such that $X\models B$ and
for every $b\in B^+\setminus\iota$, $X_i$ contains the following irregular atom $|lr_a-1\geq lr_b|$. From the fact that $\nu$ is a solution to the GCSP that includes an integer linear constraint  $lr_a-1\geq lr_b$ it
follows that inequality $lr_a^\nu-1\geq lr_b^\nu$ holds. By $\lr$ construction we conclude that $\lr(a)-1\geq\lr(b)$.
 
 By Theorem~\ref{thm:casp-ans-iff-lr}, 
$X$ is an input answer set of $\Pi$ relative to $\iota$.
\end{proof}

\begin{theorem}\label{thm:aspt-set-iff-smt-lr-sat}
	Let $\cL$ be an integer lexicon and $\cB_\cL$ be the set of all integer linear constraints over $\cL$.
	For a CASP Program (or, an ASPT(IL) program) $P=\langle \Pi,\cB_\cL,\gamma\rangle$ over vocabulary $\sigma=\sigma_r\cup\sigma_i$, and a set~$X$ of atoms from  $\sigma$,
	$X$ is an  answer set of $P$
    if and only if 
    there is a model $X\cup X_i$ of a constraint formula (or SMT(IL formula)) 
  	$$\langle IComp(\Pi,\sigma_i)\wedge  R(\Pi,\sigma_i),\cB_\cL\cup\cB_{R(\Pi,\iota)},\gamma'\rangle $$  over $\sigma_r\cup\sigma_i\cup\sigma_i^{R(\Pi,\iota)}$ so that $X_i$ is the set of atoms over $\sigma_i^{R(\Pi,\iota)}$ and 		
    $\gamma'$ is such that it coincides with $\gamma$ on atoms in $\sigma_i$ and with $\gamma^{R(\Pi,\iota)}$ on atoms 
    from $\sigma_i^{R(\Pi,\iota)}$.
\end{theorem}
	Proof of Theorem~\ref{thm:aspt-set-iff-smt-lr-sat} follows the lines of proof for  Theorem~\ref{thm:ians-set-iff-smt-lr-sat}.

Theorem~\ref{thm:aspt-set-iff-smt-lr-sat} paves the way for using SMT solvers for computing answer sets for arbitrary ASPT(IL) programs. \citeauthor{nie08}~\shortcite{nie08} introduced the notions of strong level ranking and also illustrated how strongly connected components of a dependency graph of a normal program can be used to enhance the transformation from a normal program to an SMT(DL) formula. Similar ideas could be used for enhancing the proposed translation from ASPT(IL) to SMT(IL) formalism. Such enhancements are of essence when implementation of SMT-based solver for nontight ASPT(IL) programs is considered. Implementing such enhancements is the direction of future work. 
 
\section{Conclusions}
In this paper we unified the terminology stemming from the fields of CASP and SMT solving. This unification helped us identify the special class of  uniform theories widely used in SMT practice. Given such theories, CASP and SMT solving share more in common than meets the eye. Based on this unification, we  open the doors for writing programs in the CASP formalism, while allowing SMT solving technologies to be utilized. In these settings, CASP can be seen as a possible general-purpose declarative  programming front-end for SMT technology.
In the future, we would like to investigate a similar link to a related formalism of HEX-programs~\cite{eit12}.
Overall, we expect this work to be a strong building block that will bolster the cross-fertilization between three different, even if related, automated reasoning communities: CASP, constraint (satisfaction processing) programming, and SMT.

\bibliographystyle{acmtrans}
\bibliography{abstractmods-bib}

\begin{thebibliography}{}

\bibitem[\protect\citeauthoryear{Balduccini}{Balduccini}{2009}]{bal09a}
{\sc Balduccini, M.} 2009.
\newblock Representing constraint satisfaction problems in answer set
  programming.
\newblock {\it In} Proceedings of ICLP Workshop on Answer Set Programming and
  Other Computing Paradigms (ASPOCP),
  \url{https://www.mat.unical.it/ASPOCP09/}.

\bibitem[\protect\citeauthoryear{Balduccini}{Balduccini}{2013}]{bal13a}
{\sc Balduccini, M.} 2013.
\newblock {A}{S}{P} with non-{H}erbrand partial functions: a language and
  system for practical use.
\newblock {\em Theory and Practice of Logic Programming\/}~{\em 13}, 547--561.

\bibitem[\protect\citeauthoryear{Balduccini and Lierler}{Balduccini and
  Lierler}{2017}]{lierbal16}
{\sc Balduccini, M.} {\sc and} {\sc Lierler, Y.} 2017.
\newblock Constraint answer set solver {EZCSP} and why integration schemas
  matter.
\newblock {\em Theory and Practice of Logic Programming\/}, This Issue.

\bibitem[\protect\citeauthoryear{Barrett, Conway, Deters, Hadarean,
  Jovanovi\'{c}, King, Reynolds, and Tinelli}{Barrett
  et~al\mbox{.}}{2011}]{cvc4}
{\sc Barrett, C.}, {\sc Conway, C.~L.}, {\sc Deters, M.}, {\sc Hadarean, L.},
  {\sc Jovanovi\'{c}, D.}, {\sc King, T.}, {\sc Reynolds, A.}, {\sc and} {\sc
  Tinelli, C.} 2011.
\newblock {CVC4}.
\newblock In {\em Proceedings of the 23rd International Conference on Computer
  Aided Verification (CAV'11), volume 6806 of LNCS}. Springer International
  Publishing.

\bibitem[\protect\citeauthoryear{Barrett, Fontaine, and Tinelli}{Barrett
  et~al\mbox{.}}{2015}]{smt15}
{\sc Barrett, C.}, {\sc Fontaine, P.}, {\sc and} {\sc Tinelli, C.} 2015.
\newblock {The SMT-LIB Standard: Version 2.5}.
\newblock Tech. rep., Department of Computer Science, The University of Iowa.
\newblock Available at {\tt www.SMT-LIB.org}.

\bibitem[\protect\citeauthoryear{Barrett, Sebastiani, Seshia, and
  Tinelli}{Barrett et~al\mbox{.}}{2008}]{BarretSST08}
{\sc Barrett, C.}, {\sc Sebastiani, R.}, {\sc Seshia, S.}, {\sc and} {\sc
  Tinelli, C.} 2008.
\newblock Satisfiability modulo theories.
\newblock In {\em Handbook of Satisfiability}, {A.~Biere}, {M.~Heule}, {H.~van
  Maaren}, {and} {T.~Walsch}, Eds. IOS Press, 737--797.

\bibitem[\protect\citeauthoryear{Barrett and Tinelli}{Barrett and
  Tinelli}{2014}]{BarTin-14}
{\sc Barrett, C.} {\sc and} {\sc Tinelli, C.} 2014.
\newblock Satisfiability modulo theories.
\newblock In {\em Handbook of Model Checking}, {E.~Clarke}, {T.~Henzinger},
  {and} {H.~Veith}, Eds. Springer International Publishing.

\bibitem[\protect\citeauthoryear{Bartholomew and Lee}{Bartholomew and
  Lee}{2012}]{lee12}
{\sc Bartholomew, M.} {\sc and} {\sc Lee, J.} 2012.
\newblock Stable models of formulas with intensional functions.
\newblock In {\em Proceedings of International Conference on Principles of
  Knowledge Representation and Reasoning (KR)}.

\bibitem[\protect\citeauthoryear{Bomanson, Gebser, Janhunen, Kaufmann, and
  Schaub}{Bomanson et~al\mbox{.}}{2015}]{bom15}
{\sc Bomanson, J.}, {\sc Gebser, M.}, {\sc Janhunen, T.}, {\sc Kaufmann, B.},
  {\sc and} {\sc Schaub, T.} 2015.
\newblock Answer set programming modulo acyclicity.
\newblock In {\em Proceedings of the 13th International Conference on Logic
  Programming and Nonmonotonic Reasoning, LPNMR, Lexington, KY, USA},
  {F.~Calimeri}, {G.~Ianni}, {and} {M.~Truszczynski}, Eds. Springer
  International Publishing, Cham, 143--150.

\bibitem[\protect\citeauthoryear{Calimeri, Cozza, Ianni, and Leone}{Calimeri
  et~al\mbox{.}}{2008}]{cal08}
{\sc Calimeri, F.}, {\sc Cozza, S.}, {\sc Ianni, G.}, {\sc and} {\sc Leone, N.}
  2008.
\newblock Computable functions in {A}{S}{P}: theory and implementation.
\newblock In {\em Proceedings of International Conference on Logic Programming
  ({ICLP})}. 407--424.

\bibitem[\protect\citeauthoryear{Carlsson and Fruehwirth}{Carlsson and
  Fruehwirth}{2014}]{sicstus-prolog}
{\sc Carlsson, M.} {\sc and} {\sc Fruehwirth, T.} 2014.
\newblock {\em {SICS}tus PROLOG User's Manual 4.3}.
\newblock Books On Demand - Proquest.

\bibitem[\protect\citeauthoryear{Clark}{Clark}{1978}]{cla78}
{\sc Clark, K.} 1978.
\newblock Negation as failure.
\newblock In {\em Logic and Data Bases}, {H.~Gallaire} {and} {J.~Minker}, Eds.
  Plenum Press, New York, 293--322.

\bibitem[\protect\citeauthoryear{De~Moura and Bj{\o}rner}{De~Moura and
  Bj{\o}rner}{2008}]{z3}
{\sc De~Moura, L.} {\sc and} {\sc Bj{\o}rner, N.} 2008.
\newblock Z3: An efficient {S}{M}{T} solver.
\newblock In {\em Proceedings of the Theory and Practice of Software, 14th
  International Conference on Tools and Algorithms for the Construction and
  Analysis of Systems}. 337--340.

\bibitem[\protect\citeauthoryear{Denecker and Vennekens}{Denecker and
  Vennekens}{2007}]{den07}
{\sc Denecker, M.} {\sc and} {\sc Vennekens, J.} 2007.
\newblock Well-founded semantics and the algebraic theory of non-monotone
  inductive definitions.
\newblock In {\em Proceedings of the 9th International Conference Logic
  Programming and Nonmonotonic Reasoning, {LPNMR} 2007, Tempe, AZ, USA},
  {C.~Baral}, {G.~Brewka}, {and} {J.~S. Schlipf}, Eds. Lecture Notes in
  Computer Science, vol. 4483. Springer International Publishing, 84--96.

\bibitem[\protect\citeauthoryear{Eiter, Fink, Krennwallner, and Redl}{Eiter
  et~al\mbox{.}}{2012}]{eit12}
{\sc Eiter, T.}, {\sc Fink, M.}, {\sc Krennwallner, T.}, {\sc and} {\sc Redl,
  C.} 2012.
\newblock Conflict-driven {ASP} solving with external sources.
\newblock {\em {Theory and Practice of Logic Programming}\/}~{\em 12,\/}~4-5,
  659--679.

\bibitem[\protect\citeauthoryear{Elkabani, Pontelli, and Son}{Elkabani
  et~al\mbox{.}}{2004}]{elk04}
{\sc Elkabani, I.}, {\sc Pontelli, E.}, {\sc and} {\sc Son, T.~C.} 2004.
\newblock Smodels with {CLP} and its applications: A simple and effective
  approach to aggregates in {A}{S}{P}.
\newblock In {\em Proceedings of International Conference on Logic Programming
  (ICLP)}. 73--89.

\bibitem[\protect\citeauthoryear{Erdem and Lifschitz}{Erdem and
  Lifschitz}{2001}]{erd01}
{\sc Erdem, E.} {\sc and} {\sc Lifschitz, V.} 2001.
\newblock Fages' theorem for programs with nested expressions.
\newblock In {\em Proceedings of International Conference on Logic Programming
  ({ICLP})}. 242--254.

\bibitem[\protect\citeauthoryear{Fages}{Fages}{1994}]{fag94}
{\sc Fages, F.} 1994.
\newblock Consistency of {C}lark's completion and existence of stable models.
\newblock {\em Journal of Methods of Logic in Computer Science\/}~{\em 1},
  51--60.

\bibitem[\protect\citeauthoryear{Ferraris and Lifschitz}{Ferraris and
  Lifschitz}{2005}]{fer05b}
{\sc Ferraris, P.} {\sc and} {\sc Lifschitz, V.} 2005.
\newblock Weight constraints as nested expressions.
\newblock {\em Theory and Practice of Logic Programming\/}~{\em 5}, 45--74.

\bibitem[\protect\citeauthoryear{Gebser, Janhunen, and Rintanen}{Gebser
  et~al\mbox{.}}{2014}]{geb14}
{\sc Gebser, M.}, {\sc Janhunen, T.}, {\sc and} {\sc Rintanen, J.} 2014.
\newblock {SAT} modulo graphs: acyclicity.
\newblock In {\em Proceedings of the 14th European Conference Logics in
  Artificial Intelligence, JELIA, Funchal, Madeira, Portugal}, {E.~Ferm{\'e}}
  {and} {J.~Leite}, Eds. Springer International Publishing, Cham, 137--151.

\bibitem[\protect\citeauthoryear{Gebser, Ostrowski, and Schaub}{Gebser
  et~al\mbox{.}}{2009}]{geb09}
{\sc Gebser, M.}, {\sc Ostrowski, M.}, {\sc and} {\sc Schaub, T.} 2009.
\newblock Constraint answer set solving.
\newblock In {\em Proceedings of 25th International Conference on Logic
  Programming (ICLP)}. Springer International Publishing, 235--249.

\bibitem[\protect\citeauthoryear{Gebser, Schaub, and Thiele}{Gebser
  et~al\mbox{.}}{2007}]{geb07b}
{\sc Gebser, M.}, {\sc Schaub, T.}, {\sc and} {\sc Thiele, S.} 2007.
\newblock Gringo: A new grounder for answer set programming.
\newblock In {\em Proceedings of the Ninth International Conference on Logic
  Programming and Nonmonotonic Reasoning}. 266--271.

\bibitem[\protect\citeauthoryear{Gelfond and Przymusinska}{Gelfond and
  Przymusinska}{1996}]{gel96}
{\sc Gelfond, M.} {\sc and} {\sc Przymusinska, H.} 1996.
\newblock Towards a theory of elaboration tolerance: Logic programming
  approach.
\newblock {\em International Journal of Software Engineering and Knowledge
  Engineering\/}~{\em 6,\/}~1, 89--112.

\bibitem[\protect\citeauthoryear{IBM}{IBM}{2009}]{cplex}
IBM 2009.
\newblock {\em IBM ILOG AMPL Version 12.1 User's Guide}.
\newblock IBM.
\newblock
  \url{http://www.ibm.com/software/commerce/optimization/cplex-optimizer/}.

\bibitem[\protect\citeauthoryear{Janhunen, Liu, and Niemela}{Janhunen
  et~al\mbox{.}}{2011}]{jan11}
{\sc Janhunen, T.}, {\sc Liu, G.}, {\sc and} {\sc Niemela, I.} 2011.
\newblock Tight integration of non-ground answer set programming and
  satisfiability modulo theories.
\newblock In {\em Proceedings of the 1st Workshop on Grounding and
  Transformations for Theories with Variables}.

\bibitem[\protect\citeauthoryear{King, Barrett, and Tinelli}{King
  et~al\mbox{.}}{2014}]{king14}
{\sc King, T.}, {\sc Barrett, C.}, {\sc and} {\sc Tinelli, C.} 2014.
\newblock Leveraging linear and mixed integer programming for {SMT}.
\newblock In {\em Proceedings of the 14th Conference on Formal Methods in
  Computer-Aided Design}. FMCAD '14. FMCAD Inc, Austin, TX, 24:139--24:146.

\bibitem[\protect\citeauthoryear{Lee and Meng}{Lee and Meng}{2013}]{lee13}
{\sc Lee, J.} {\sc and} {\sc Meng, Y.} 2013.
\newblock Answer set programming modulo theories and reasoning about continuous
  changes.
\newblock In {\em Proceedings of the 23rd International Joint Conference on
  Artificial Intelligence (IJCAI-13), Beijing, China, August 3-9, 2013}.

\bibitem[\protect\citeauthoryear{Lierler}{Lierler}{2014}]{lier14}
{\sc Lierler, Y.} 2014.
\newblock Relating constraint answer set programming languages and algorithms.
\newblock {\em Artificial Intelligence\/}~{\em 207C}, 1--22.

\bibitem[\protect\citeauthoryear{Lierler and Susman}{Lierler and
  Susman}{2016}]{sus16}
{\sc Lierler, Y.} {\sc and} {\sc Susman, B.} 2016.
\newblock Constraint answer set programming versus satisfiability modulo
  theories.
\newblock In {\em Proceedings of the 25th International Joint Conference on
  Artificial Intelligence (IJCAI)}. 1181--1187.

\bibitem[\protect\citeauthoryear{Lierler and Truszczynski}{Lierler and
  Truszczynski}{2011}]{lt2011}
{\sc Lierler, Y.} {\sc and} {\sc Truszczynski, M.} 2011.
\newblock Transition systems for model generators --- a unifying approach.
\newblock {\em Theory and Practice of Logic Programming, 27th International
  Conference on Logic Programming (ICLP'11) Special Issue\/}~{\em 11, issue
  4-5}.

\bibitem[\protect\citeauthoryear{Lifschitz}{Lifschitz}{2012}]{lif12a}
{\sc Lifschitz, V.} 2012.
\newblock Logic programs with intensional functions.
\newblock In {\em Proceedings of International Conference on Principles of
  Knowledge Representation and Reasoning (KR)}.

\bibitem[\protect\citeauthoryear{Lifschitz, Tang, and Turner}{Lifschitz
  et~al\mbox{.}}{1999}]{lif99d}
{\sc Lifschitz, V.}, {\sc Tang, L.~R.}, {\sc and} {\sc Turner, H.} 1999.
\newblock Nested expressions in logic programs.
\newblock {\em Annals of Mathematics and Artificial Intelligence\/}~{\em 25},
  369--389.

\bibitem[\protect\citeauthoryear{Liu, Janhunen, and Niemela}{Liu
  et~al\mbox{.}}{2012}]{liu12}
{\sc Liu, G.}, {\sc Janhunen, T.}, {\sc and} {\sc Niemela, I.} 2012.
\newblock Answer set programming via mixed integer programming.
\newblock In {\em Proceedings of the 13th International Conference on
  Principles of Knowledge Representation and Reasoning (KR)}.

\bibitem[\protect\citeauthoryear{Marriott and Stuckey}{Marriott and
  Stuckey}{1998}]{mar98}
{\sc Marriott, K.} {\sc and} {\sc Stuckey, P.~J.} 1998.
\newblock {\em Programming with Constraints: An Introduction}.
\newblock MIT Press.

\bibitem[\protect\citeauthoryear{Mellarkod, Gelfond, and Zhang}{Mellarkod
  et~al\mbox{.}}{2008}]{Mellarkod2009}
{\sc Mellarkod, V.~S.}, {\sc Gelfond, M.}, {\sc and} {\sc Zhang, Y.} 2008.
\newblock Integrating answer set programming and constraint logic programming.
\newblock {\em Annals of Mathematics and Artificial Intelligence\/}~{\em
  53,\/}~1, 251--287.

\bibitem[\protect\citeauthoryear{Niemel{\"a}}{Niemel{\"a}}{2008}]{nie08}
{\sc Niemel{\"a}, I.} 2008.
\newblock Stable models and difference logic.
\newblock {\em Annals of Mathematics and Artificial Intelligence\/}~{\em 53},
  313--329.

\bibitem[\protect\citeauthoryear{Niemel{\"a} and Simons}{Niemel{\"a} and
  Simons}{2000}]{nie00}
{\sc Niemel{\"a}, I.} {\sc and} {\sc Simons, P.} 2000.
\newblock Extending the {Smodels} system with cardinality and weight
  constraints.
\newblock In {\em Logic-Based Artificial Intelligence}, {J.~Minker}, Ed.
  Kluwer, 491--521.

\bibitem[\protect\citeauthoryear{Nieuwenhuis and Oliveras}{Nieuwenhuis and
  Oliveras}{2005}]{nie05}
{\sc Nieuwenhuis, R.} {\sc and} {\sc Oliveras, A.} 2005.
\newblock {DPLL(T)} with exhaustive theory propagation and its application to
  difference logic.
\newblock In {\em Proceedings of the 17th International Conference on Computer
  Aided Verification (CAV'05), volume 3576 of LNCS}. Springer International
  Publishing.

\bibitem[\protect\citeauthoryear{Oikarinen and Janhunen}{Oikarinen and
  Janhunen}{2006}]{oik06}
{\sc Oikarinen, E.} {\sc and} {\sc Janhunen, T.} 2006.
\newblock Modular equivalence for normal logic programs.
\newblock In {\em Proceedings of the 17th European Conference on Artificial
  Intelligence, ECAI 2006}, {G.~Brewka}, {S.~Coradeschi}, {A.~Perini}, {and}
  {P.~Traverso}, Eds. IOS Press, Amsterdam, The Netherlands, 412--416.

\bibitem[\protect\citeauthoryear{Susman and Lierler}{Susman and
  Lierler}{2016}]{sus16b}
{\sc Susman, B.} {\sc and} {\sc Lierler, Y.} 2016.
\newblock {S}{M}{T}-based constraint answer set solver {E}{Z}{S}{M}{T} (system
  description).
\newblock In {\em Proceedings of 32th International Conference on Logic
  Programming (ICLP)}. Dagstuhl Publishing, OpenAccess Series in Informatics
  (OASIcs).

\bibitem[\protect\citeauthoryear{Wielemaker, Schrijvers, Triska, and
  Lager}{Wielemaker et~al\mbox{.}}{2012}]{swi-prolog}
{\sc Wielemaker, J.}, {\sc Schrijvers, T.}, {\sc Triska, M.}, {\sc and} {\sc
  Lager, T.} 2012.
\newblock {SWI-Prolog}.
\newblock {\em Theory and Practice of Logic Programming\/}~{\em 12,\/}~1-2,
  67--96.

\bibitem[\protect\citeauthoryear{Zhou}{Zhou}{2012}]{b-prolog}
{\sc Zhou, N.} 2012.
\newblock The language features and architecture of {B}-{P}rolog.
\newblock {\em Theory and Practice of Logic Programming\/}~{\em 12,\/}~1-2
  (Jan.), 189--218.

\end{thebibliography}

\end{document}